\documentclass[amsmath,amssymb, aps, 11pt]{revtex4-2}

\usepackage{enumerate}
\usepackage{physics}
\usepackage{mathrsfs}
\usepackage{amsthm}
\usepackage{bm} 
\usepackage{listings}

\newtheorem{theorem}{Theorem}[section]
\newtheorem{proposition}[theorem]{Proposition}
\newtheorem{lemma}[theorem]{Lemma}

\newtheorem{assumption}[theorem]{Assumption}
\newtheorem{remark}[theorem]{Remark}
\usepackage{tikz,xcolor}
\usepackage{graphicx}
\usepackage{dcolumn}
\usepackage{bm}
\usepackage{hyperref}
\usepackage[title]{appendix}
\usepackage{physics}
\hypersetup{
    colorlinks=true,
    urlcolor= blue,
    citecolor=blue,
    linkcolor= blue}

\definecolor{lime}{HTML}{A6CE39}
\DeclareRobustCommand{\orcidicon}{%
	\begin{tikzpicture}
	\draw[lime, fill=lime] (0,0) 
	circle [radius=0.16] 
	node[white] {{\fontfamily{qag}\selectfont \tiny ID}};
	\draw[white, fill=white] (-0.0625,0.095) 
	circle [radius=0.007];
	\end{tikzpicture}
	\hspace{-2mm}
}

\foreach \x in {A, ..., Z}{%
	\expandafter\xdef\csname orcid\x\endcsname{\noexpand\href{https://orcid.org/\csname orcidauthor\x\endcsname}{\noexpand\orcidicon}}
}

\begin{document}

\preprint{APS/123-QED}

\title{Explicit Analytical Representations for the Schwarzschild Radial Equation\\
	via Hypergeometric and Frobenius Expansions}

\author{J. G. R. Valangelis\orcidA{}}%
 \email{guilhermexitado@ufpi.edu.br}
\affiliation{Departamento de Física, Universidade Federal do Piauí, Teresina, 64049-550, Piauí, Brazil}

\author{ Ailton C. Nascimento\orcidB{}}
\email{ailton.nascimento@ufpi.edu.br}
 \affiliation{Departamento de Matemática, Universidade Federal do Piauí, Teresina, 64049-550, Piauí, Brazil}

\author{Helder~A.~S.~Costa\orcidC{}}%
 \email{hascosta@ufpi.edu.br}
\affiliation{Departamento de Física, Universidade Federal do Piauí, Teresina, 64049-550, Piauí, Brazil}

\begin{abstract}
The massive Klein-Gordon equation on the Schwarzschild exterior reduces to a radial differential equation of confluent-Heun type. We present its exact reduction by rigorously retaining a subleading centrifugal term that preserves the horizon indicial exponents but critically shifts the accessory parameters. Within the Svartholm--Schmidt hypergeometric-expansion framework, we derive the three-term recurrence relation and establish the corrected continued-fraction compatibility condition matching the minimal tail to the lower-end recurrence. Crucially, we demonstrate that the physical compactification coordinate $z=(r-1)/r$ shifts the irregular singular point to $z=1$, inherently transforming the hypergeometric expansion into a five-term recurrence that cannot degenerate to three terms. To circumvent this obstruction, we construct the horizon-normalized physical branch via a direct Frobenius series at $z=0$. This representation exhibits high-order convergence and near-precision-floor residuals against the exact radial ODE, while clarifying why quasinormal-mode spectral selection remains a connection problem at the irregular singular endpoint.
\end{abstract}

\maketitle

\newpage

\tableofcontents

    \section{Introduction}

The propagation of a test field on a Schwarzschild background reduces, upon separation of variables, to a single linear ordinary differential equation for the radial profile. For a massive scalar field, the analytically extended radial equation has two regular singular points, at $r=0$ and the event horizon at $r=2M$, together with a rank-one irregular singular point at spatial infinity ($r \to \infty$). Consequently, the governing radial dynamics falls strictly within the class of the \emph{confluent Heun equation} \cite{RonveauxBook,SlavyanovLayBook,Decarreau1978,Fiziev2006}. Equations of the Heun family generally resist closed-form solutions in terms of classical hypergeometric functions; hence, constructing explicit, highly convergent representations is essential for spectral and scattering problems in black-hole perturbation theory, building on the classical analyses of Regge, Wheeler, and Zerilli \cite{ReggeWheeler1957,Zerilli1970}, with subsequent developments covering quasinormal modes and related spectral problems \cite{Berti2009,KonoplyaZhidenkoReview2011,Jansen2017}.

A foundational framework for analytic expansions of Heun solutions was developed through the works of Schmidt and Svartholm \cite{Schmidt1932,Svartholm1939}; see also \cite{RonveauxBook}. In this scheme, the radial amplitude is expanded over a basis of Gauss hypergeometric functions adapted to a regular singular endpoint. The expansion coefficients satisfy a three-term recurrence relation whose minimal solution is characterized by Perron's theory \cite{Perron1929}, with computational aspects of such recurrence relations thoroughly detailed by Gautschi \cite{Gautschi1967}. Matching this minimal tail to the lower-end recurrence yields an infinite continued-fraction compatibility condition, which underlies the continued-fraction approach developed by Leaver for black-hole perturbation spectra \cite{Leaver1985,Leaver1986}. Crucially, however, this algebraic matching does not automatically impose the required asymptotics at the irregular endpoint, which remains a fundamentally separate connection problem \cite{Wasow1965}.

In this work, we revisit the analytic structure of the massive scalar field on the Schwarzschild metric and clarify several subtleties in its Heun reduction and spectral representations. Specifically, this paper makes four main contributions:

\begin{enumerate}
    \item \textbf{Exact Confluent Heun Reduction and Accessory Parameters:} We derive the explicit confluent Heun reduction of the separated radial equation with fully corrected accessory parameters. In particular, we rigorously retain the subleading centrifugal/metric term $-1/[r^2(r-1)]$, which, while leaving the indicial exponents invariant, significantly shifts the accessory parameter and is indispensable for quantitative spectral computations.
    
    \item \textbf{Structured Svartholm--Schmidt Theory:} We reformulate the Svartholm--Schmidt expansion for the canonical confluent Heun operator in a systematic algebraic hierarchy: establishing the diagonal eigenrelation, the contiguous relations, the tridiagonal matrix reduction, and the resulting three-term recurrence. We strictly decouple the Perron minimal-tail asymptotics from the lower-end boundary match, yielding the exact continued-fraction selection equation.
    
    \item \textbf{Recurrence Structure under Physical Compactification:} We demonstrate that under the standard physical compactification coordinate $z=(r-1)/r$, where the irregular singularity maps to $z=1$, the reduced equation departs from the canonical single-pole form. Projecting onto the standard hypergeometric basis unavoidably generates a \emph{five-term} recurrence relation rather than a three-term one. We provide the closed-form operator identities and extreme bands, proving that this recurrence cannot truncate or degenerate to three terms.
    
    \item \textbf{Horizon Frobenius Branch and Numerical Validation:} Addressing the five-term obstruction, we construct the physical branch via a direct Frobenius series expansion around the event horizon. This provides an explicit, fast-converging local analytic solution on the domain $0 \le z < 1$, validated against high-precision numerical integration of the radial equation on compact exterior intervals.
\end{enumerate}

The remainder of this paper is structured as follows. 
In Sections~\ref{sec:reduction}--\ref{sec:compactified}, we carry out the separation of the massive Klein-Gordon equation on the Schwarzschild background, derive the exact confluent-Heun reduction with corrected accessory parameters, and introduce the compactified radial coordinate $z = (r-1)/r$. 
Sections~\ref{sec:svartholm}--\ref{sec:matching} are devoted to the foundational Svartholm-Schmidt framework for the canonical confluent-Heun operator: we establish the hypergeometric basis eigenrelation, perform the tridiagonal matrix reduction leading to the three-term recurrence, analyze the Perron minimal-tail asymptotics, and derive the corrected continued-fraction compatibility condition alongside its irregular connection problem. 
In Section~\ref{sec:fiveterm}, we analyze the physical compactified equation, establishing the necessity of the five-term recurrence relation and demonstrating why it cannot degenerate to three terms. 
Sections~\ref{sec14}--\ref{sec:plot-interpretation} present the explicit horizon-normalized Frobenius construction, its implementation in \textsc{Mathematica}, high-precision residual diagnostics against the exact radial ODE, and an analytical discussion on why the five-term recurrence serves as a finite connection object rather than a direct quasinormal-mode solver, complemented by a benchmark against standard Frobenius--Leaver continued fractions. 
Finally, Section~\ref{sec:closing} summarizes our conclusions and future outlook, while the complete, self-contained \textsc{Mathematica} routine is provided in Appendix~\ref{app:code}.
	
	\section{Reduction of the Schwarzschild radial equation to confluent-Heun form}\label{sec:radial}
	\label{sec:reduction}
    We consider the massive Klein-Gordon equation for a scalar test field propagating on the Schwarzschild exterior curved geometry \cite{BirrellDavies1982,Wald1984} (in units $2M=1$), written (after multiplication by $r^2$) as
	\begin{equation}\label{eq:SchwarzschildKG_PDE}
		-\frac{r^{2}}{1-\frac{1}{r}}\,\partial_{t}^{2}\psi
		+\partial_{r}\!\Bigl((r^{2}-r)\,\partial_{r}\psi\Bigr)
		+\Delta_{\mathbb S^{2}}\psi
		-m^{2}r^{2}\psi=0,
	\end{equation}
	where $\Delta_{\mathbb S^{2}}$ is the Laplace-Beltrami operator on the unit sphere,
	\[
	\Delta_{\mathbb S^{2}}
	=\frac{1}{\sin\theta}\partial_{\theta}\bigl(\sin\theta\,\partial_{\theta}\bigr)
	+\frac{1}{\sin^{2}\theta}\partial_{\varphi}^{2}.
	\]
	(If \eqref{eq:SchwarzschildKG_PDE} is written with $\partial_\varphi^2$ only, it is
	understood that the standard $\sin^{-2}\theta$ factor is included in the definition
	of the angular operator.)
	
	Let $Y_{\ell n}(\theta,\varphi)$ be spherical harmonics,
	$\Delta_{\mathbb S^{2}}Y_{\ell n}=-\ell(\ell+1)Y_{\ell n}$, and set
	\begin{equation}\label{eq:sep_ansatz}
		\psi(t,r,\theta,\varphi)=e^{-i\omega t}\,Y_{\ell n}(\theta,\varphi)\,\frac{R(r)}{r}.
	\end{equation}
	With $f(r)=1-\frac1r=\frac{r-1}{r}$, substitution of \eqref{eq:sep_ansatz} into
	\eqref{eq:SchwarzschildKG_PDE} yields the standard radial ODE
	\begin{equation}\label{eq:RadialOriginal_revised}
		\frac{d}{dr}\!\Bigl(f(r)\,R'(r)\Bigr)
		+\Bigl(\frac{\omega^{2}}{f(r)}-\frac{\ell(\ell+1)}{r^{2}}-m^{2}-\frac{1}{r^{3}}\Bigr)R(r)=0,
		\qquad r\in(1,\infty).
	\end{equation}
	Equivalently, in monic form,
	\begin{equation}\label{eq:RadialMonic_revised}
		R''(r)+\frac{1}{r(r-1)}R'(r)
		+\left(\frac{\omega^{2}r^{2}}{(r-1)^{2}}-\frac{\ell(\ell+1)}{r(r-1)}-\frac{m^{2}r}{r-1}-\frac{1}{r^{2}(r-1)}\right)R(r)=0.
	\end{equation}
	The last term $-1/(r^{2}(r-1))$ in the potential of \eqref{eq:RadialMonic_revised}
	(equivalently the term $-1/r^{3}$ in \eqref{eq:RadialOriginal_revised}) originates from the
	$1/r$ factor in the ansatz \eqref{eq:sep_ansatz}: the operator
	$\partial_r\!\bigl((r^2-r)\partial_r(R/r)\bigr)$ generates a contribution $-f'(r)R/r=-(1/r^{3})R$.
	It is subleading at both finite singular points (it does not alter the indicial
	exponents $\rho=\pm i\omega$ at $r=1$, nor the asymptotic relation $\kappa^{2}=m^{2}-\omega^{2}$
	at $r=\infty$), but it must be retained for the accessory parameters and for any quantitative
	solution to be correct.
	
	\paragraph{Singularity structure.}
	Equation \eqref{eq:RadialMonic_revised} has regular singular points at $r=0$ and $r=1$.
	The point $r=\infty$ is \emph{irregular} whenever $\omega^{2}\neq m^{2}$, reflecting the
	exponential/oscillatory asymptotics at spatial infinity. Consequently, the natural Heun
	class for \eqref{eq:RadialMonic_revised} is the \emph{confluent Heun family} (two regular
	finite singularities and one irregular singularity at infinity), rather than the general
	Heun family.
	
	\section{A confluent-Heun normalization}\label{sec:normalization}
	
	To expose the confluent-Heun structure and to prepare an explicit hypergeometric
	expansion, we factor out the canonical local behaviors at the horizon and at infinity.
	Let $\kappa$ be defined by
	\begin{equation}\label{eq:kappa_def}
		\kappa^{2}=m^{2}-\omega^{2},
	\end{equation}
	where one fixes a branch according to the desired boundary condition at $r=\infty$
	(decaying, outgoing, etc.). Moreover, the horizon exponents of \eqref{eq:RadialMonic_revised}
	are $\rho=\pm i\omega$, corresponding to ingoing/outgoing modes at $r=1$.
	
	\begin{proposition}[Confluent-Heun reduction]\label{prop:HeunC_reduction}
		Fix $\rho\in\{+i\omega,-i\omega\}$ and $\kappa$ as in \eqref{eq:kappa_def}. Define
		\begin{equation}\label{eq:gauge_transform}
			R(r)=e^{\kappa r}(r-1)^{\rho}\,H(r).
		\end{equation}
		Then $H$ satisfies a confluent-Heun-type equation of the form
		\begin{equation}\label{eq:HeunC_general_form}
			H''+\left(2\kappa+\frac{1}{r(r-1)}+\frac{2\rho}{r-1}\right)H'
			+\left(\frac{A_{0}}{r^{2}}+\frac{A_{1}}{r}+\frac{B_{1}}{r-1}\right)H=0,
		\end{equation}
		where the coefficients are explicit:
		\begin{align}
			A_{0}&=1, \label{eq:A0}\\
			A_{1}&=\ell(\ell+1)\;-\;\kappa\;+\;\rho\;+\;1,\label{eq:A1}\\
			B_{1}&=-\,\ell(\ell+1)\;+\;2\kappa\rho\;+\;\kappa\;-\;m^{2}\;-\;\rho\;+\;2\omega^{2}\;-\;1.\label{eq:B1}
		\end{align}
		In particular, \eqref{eq:HeunC_general_form} has regular singularities at $r=0$ and $r=1$
		and an irregular singularity at $r=\infty$, hence belongs to the confluent-Heun class.
	\end{proposition}
	
	\begin{proof}
		Insert \eqref{eq:gauge_transform} into \eqref{eq:RadialMonic_revised}, expand $R',R''$,
		and simplify. The choice $\rho=\pm i\omega$ cancels the $(r-1)^{-2}$ term in the reduced
		equation, ensuring that the horizon singularity is regular for $H$.
		The exponential factor $e^{\kappa r}$ isolates the dominant asymptotics at infinity and
		produces the constant term $2\kappa$ in the coefficient of $H'$, which is the hallmark of
		the irregular singularity at $r=\infty$ in confluent-Heun normalizations. Collecting the
		remaining terms yields \eqref{eq:HeunC_general_form} with coefficients
		\eqref{eq:A0}-\eqref{eq:B1}. We emphasize that $A_{0}=1\neq0$: the genuine $1/r^{2}$ term
		in the coefficient of $H$ is precisely the image of the centrifugal correction
		$-1/(r^{2}(r-1))$ in \eqref{eq:RadialMonic_revised}. The horizon exponents $\rho=\pm i\omega$
		are unchanged, since the dropped term is subleading at $r=1$.
	\end{proof}

One may also introduce an additional factor at the nonphysical singular point $r=0$ by
writing
\[
R(r)=e^{\kappa r}(r-1)^{\rho}r^{\sigma}H(r).
\]
The parameter $\sigma$ can then be chosen to match a preferred canonical normalization,
for instance the Maple/Mathematica \texttt{HeunC} convention. In the present exterior
problem, however, the physical domain is $r\in(1,\infty)$, and the minimal gauge
\eqref{eq:gauge_transform} already isolates the relevant horizon and infinity
behaviors. We therefore keep this normalization, which also leads to a simpler
hypergeometric and Frobenius analysis below.

	\section{Compactified coordinate and the physical reduced equation}\label{sec:compactified}
	
	For the hypergeometric expansion it is advantageous to work in the compact coordinate
	\begin{equation}\label{eq:z_change_revised}
		z=\frac{r-1}{r}=1-\frac1r,
		\qquad r=\frac{1}{1-z},
		\qquad r\in[1,\infty)\Longleftrightarrow z\in[0,1).
	\end{equation}
	Note that $z=0$ is the horizon and $z\to1^{-}$ corresponds to spatial infinity.
	
	\begin{lemma}\label{lem:derivative_change_revised}
		Under \eqref{eq:z_change_revised} one has
		\[
		\frac{d}{dr}=(1-z)^{2}\frac{d}{dz},
		\qquad
		\frac{d^{2}}{dr^{2}}=(1-z)^{4}\frac{d^{2}}{dz^{2}}-2(1-z)^{3}\frac{d}{dz}.
		\]
	\end{lemma}
	
	\begin{proof}
		Since $z=1-\frac1r$, we have $dz/dr=1/r^{2}=(1-z)^{2}$, hence $\frac{d}{dr}=(dz/dr)\frac{d}{dz}$.
		Differentiating once more gives the stated identity for $\frac{d^{2}}{dr^{2}}$.
	\end{proof}
	
	Combining Proposition \ref{prop:HeunC_reduction} with Lemma \ref{lem:derivative_change_revised}
	transports \eqref{eq:HeunC_general_form} into a second-order ODE on $z\in[0,1)$ with a regular
	singularity at $z=0$ and an irregular singularity at $z=1$. Explicitly, clearing denominators,
	the reduced equation reads
	\begin{equation}\label{eq:zform_correct}
		z(z-1)^{3}\,H''(z)+c_{1}(z)\,H'(z)+c_{0}(z)\,H(z)=0,
		\qquad z\in[0,1),
	\end{equation}
	with the polynomial coefficients
	\begin{align}
		c_{1}(z)&=3z^{3}+(2\kappa-2\rho-7)z^{2}+(-2\kappa+4\rho+5)z-(2\rho+1),\label{eq:zform_c1}\\
		c_{0}(z)&=z^{2}+\bigl(-\ell(\ell+1)+\kappa-\rho-2\bigr)z
		+\bigl(\ell(\ell+1)-2\kappa\rho-\kappa+m^{2}-2\omega^{2}+\rho+1\bigr).\label{eq:zform_c0}
	\end{align}
	Note that the leading coefficient $z(z-1)^{3}$ vanishes to third order at $z=1$. Dividing
	\eqref{eq:zform_correct} by $z(z-1)$ to expose the drift/potential, the coefficient of $H'$
	acquires a \emph{double} pole and the coefficient of $H$ a \emph{triple} pole at $z=1$,
	confirming that $z=1$ (the image of spatial infinity) is an \emph{irregular} singular point.
	At $z=0$ the equation remains regular singular with indicial equation $s(s+2\rho)=0$,
	consistent with the choice \eqref{eq:gamma_choice} below, $\gamma=1+2\rho$.
	
	\begin{equation}\label{eq:gamma_choice}
		\gamma=1+2\rho .
	\end{equation}
	
	\begin{remark}[Standard form versus physical equation]\label{rem:standard_vs_physical}
		Equation \eqref{eq:zform_correct} is \emph{not} of the simple-pole confluent-Heun
		standard form \eqref{eq:HeunC_Schw} studied in Sections \ref{sec:SS}-\ref{sec:recurrence};
		the latter places its irregular point at $z=\infty$, whereas the physical reduction places
		it at $z=1$. This single fact organizes the rest of the paper. The Gauss-hypergeometric
		(Jacobi-type) expansion developed in Part II is therefore the Svartholm-Schmidt theory
		for the \emph{abstract} standard form; it produces the classical three-term recurrence.
		The \emph{physical} Schwarzschild equation \eqref{eq:zform_correct}, expanded in the same
		hypergeometric basis, obeys instead a \emph{five-term} recurrence
		(Theorem \ref{thm:SS_fiveterm}), and its horizon-normalized branch is constructed directly
		from \eqref{eq:zform_correct} by the Frobenius method of Section \ref{sec14}, whose accuracy
		we validate against the original radial equation \eqref{eq:RadialMonic_revised} on compact
		exterior intervals in
		Sections\ref{sec:interpret-output}-\ref{sec:plot-interpretation}. The single datum that
		carries over from \eqref{eq:zform_correct} to \eqref{eq:HeunC_Schw} is the residue at $z=0$,
		$\Gamma=\gamma=1+2\rho$.
	\end{remark}

	\section{Svartholm-Schmidt theory for the standard confluent-Heun operator}
    \label{sec:svartholm}
	
	\subsection{The standard confluent-Heun operator and its hypergeometric basis}\label{sec:SS}
	
	Throughout Part II we work with the \emph{abstract} confluent-Heun equation in the
	standard form
	\begin{equation}\label{eq:HeunC_Schw}
		H''(z)
		+\left(\alpha_{\mathrm H}+\frac{\Gamma}{z}+\frac{\Delta}{z-1}\right)H'(z)
		+\left(\frac{\mu}{z}+\frac{\eta}{z-1}\right)H(z)=0,
		\qquad z\in(0,1),
	\end{equation}
	with the five parameters
	\begin{equation}\label{eq:HeunC_params}
		\bigl(\alpha_{\mathrm H},\Gamma,\Delta,\mu,\eta\bigr).
	\end{equation}

	\begin{assumption}[Generic hypergeometric parameters]\label{ass:generic_SS}
		Throughout the Svartholm-Schmidt reduction we work away from the exceptional parameter
		values at which the hypergeometric basis or the contiguous coefficients degenerate. More
		precisely, we assume
		\[
		\Gamma\notin\{0,-1,-2,\ldots\},
		\]
		and, for every index $n\ge0$ used in the expansion,
		\[
		2n+\Omega\neq0,\qquad 2n+\Omega+1\neq0,\qquad 2n+\Omega-1\neq0,
		\]
		where $\Omega=\Gamma+\Delta-1$. We also exclude, when using the lower-end recurrence and
		continued fractions, the isolated parameter values for which the corresponding recurrence
		denominators vanish. Exceptional cases may be treated separately by limiting or
		renormalized bases, but they are not part of the generic Svartholm-Schmidt construction
		developed below.
	\end{assumption}

	The form \eqref{eq:HeunC_Schw} has simple poles at $z=0,1$ and its irregular point at
	$z=\infty$. As stressed in Remark \ref{rem:standard_vs_physical}, the physical
	Schwarzschild reduction \eqref{eq:zform_correct} is \emph{not} of this form; the
	development of this Part therefore establishes the Svartholm-Schmidt machinery for the
	standard model \eqref{eq:HeunC_Schw} per se, with $\Gamma=1+2\rho$ inherited from the
	horizon data, while the horizon-normalized physical branch is built in Part III.
	
	\subsection{Operator formulation}\label{subsec:SS_operator}
	
	Multiplying \eqref{eq:HeunC_Schw} by $z(z-1)$ we write the equation as
	\begin{equation}\label{eq:HeunC_operator_form}
		\mathcal{L}H=0,
		\qquad
		\mathcal{L}:=\Lambda_1+\alpha_{\mathrm H}\Lambda_2+\bigl(\mu(z-1)+\eta z\bigr),
	\end{equation}
	where the differential operators $\Lambda_1$ and $\Lambda_2$ are
	\begin{equation}\label{eq:Lambda_defs}
		\Lambda_1:=z(z-1)\frac{d^2}{dz^2}+\bigl(\Gamma(z-1)+\Delta z\bigr)\frac{d}{dz},
		\qquad
		\Lambda_2:=z(z-1)\frac{d}{dz}.
	\end{equation}
	Thus $\mathcal{L}$ is built from $\Lambda_1$, $\Lambda_2$, multiplication by $z$, and a
	scalar; this algebraic structure is exactly what will make the Svartholm-Schmidt
	expansion tridiagonal.
	
	\subsection{Hypergeometric basis and eigenrelation}\label{subsec:SS_basis}
	
	To avoid collision with the physical frequency $\omega$, set
	\begin{equation}\label{eq:Omega_def}
		\Omega:=\Gamma+\Delta-1,
	\end{equation}
	the \emph{method parameter} of the hypergeometric basis. For each index
	$n\in\mathbb{N}_0$ define
	\begin{equation}\label{eq:yn_def}
		y_n(z):={}_2F_1\bigl(-n,\;n+\Omega;\;\Gamma;\;z\bigr).
	\end{equation}
	Under Assumption \ref{ass:generic_SS}, this is a polynomial of degree $n$ and the family
	$\{y_n\}_{n\ge0}$ is a graded basis of $\mathbb{C}[z]$.

	\begin{lemma}[Eigenrelation for $\Lambda_1$]\label{lem:Lambda1_eig}
		The basis functions \eqref{eq:yn_def} satisfy
		\begin{equation}\label{eq:Lambda1eig}
			\Lambda_1 y_n=n(n+\Omega)\,y_n,
			\qquad n\in\mathbb{N}_0.
		\end{equation}
	\end{lemma}
	
	\begin{proof}
		The Gauss function $u={}_2F_1(a,b;c;z)$ solves
		$z(1-z)u''+\bigl(c-(a+b+1)z\bigr)u'-ab\,u=0$. With $a=-n$, $b=n+\Omega$, $c=\Gamma$ one
		has $a+b+1=\Omega+1$ and $ab=-n(n+\Omega)$, so
		\begin{equation}\label{eq:Gauss_yn}
			z(1-z)y_n''+\bigl(\Gamma-(\Omega+1)z\bigr)y_n'+n(n+\Omega)y_n=0 .
		\end{equation}
		Since $z(z-1)=-z(1-z)$, \eqref{eq:Gauss_yn} gives
		$z(z-1)y_n''=\bigl(\Gamma-(\Omega+1)z\bigr)y_n'+n(n+\Omega)y_n$. Substituting into
		$\Lambda_1y_n=z(z-1)y_n''+\bigl(\Gamma(z-1)+\Delta z\bigr)y_n'$ yields
		\[
		\Lambda_1y_n=n(n+\Omega)y_n+\bigl(\Gamma+\Delta-\Omega-1\bigr)z\,y_n'
		=n(n+\Omega)y_n,
		\]
		because $\Omega=\Gamma+\Delta-1$ forces the bracket to vanish.
	\end{proof}
	
	\begin{remark}\label{rem:Jacobi}
		For $n\in\mathbb{N}_0$ the function $y_n$ is a polynomial of degree $n$, proportional to
		a Jacobi polynomial,
		\[
		y_n(z)\propto P_n^{(\Omega-\Gamma,\;\Gamma-1)}(2z-1).
		\]
		In particular $\{y_n\}_{n\ge0}$ is a graded basis of $\mathbb{C}[z]$ (one polynomial of
		each degree), so every polynomial of degree $\le N$ has a unique expansion in
		$\{y_0,\dots,y_N\}$.
	\end{remark}
	
	\section{Three-term relations and the tridiagonal reduction}\label{sec:tridiagonal}
	
	The next lemma is the algebraic heart of the method: multiplication by $z$ and the
	first-order operator $\Lambda_2$ preserve the span of three consecutive basis elements.
	
	\begin{lemma}[Three-term expansions]\label{lem:SS_3term}
		For each $n\in\mathbb{N}_0$,
		\begin{align}
			z\,y_n&=A_n\,y_{n+1}+B_n\,y_n+C_n\,y_{n-1},\label{eq:z_three}\\[3pt]
			\Lambda_2 y_n&=A'_n\,y_{n+1}+B'_n\,y_n+C'_n\,y_{n-1},\label{eq:Lambda2_three}
		\end{align}
		where
		\begin{equation}\label{eq:ABC_coeffs}
			\begin{aligned}
				A_n&=-\frac{(n+\Omega)(n+\Gamma)}{(2n+\Omega)(2n+\Omega+1)},\\[3pt]
				B_n&=\frac{2n(n+\Omega)+\Gamma(\Omega-1)}{(2n+\Omega+1)(2n+\Omega-1)},\\[3pt]
				C_n&=-\frac{n(n+\Delta-1)}{(2n+\Omega)(2n+\Omega-1)},
			\end{aligned}
			\qquad
			\begin{aligned}
				A'_n&=-\frac{n(n+\Omega)(n+\Gamma)}{(2n+\Omega)(2n+\Omega+1)},\\[3pt]
				B'_n&=\frac{n(n+\Omega)(\Gamma-\Delta)}{(2n+\Omega+1)(2n+\Omega-1)},\\[3pt]
				C'_n&=\frac{n(n+\Omega)(n+\Delta-1)}{(2n+\Omega)(2n+\Omega-1)}.
			\end{aligned}
		\end{equation}
		In particular $A'_n=n\,A_n$ and $C'_n=-(n+\Omega)\,C_n$.
	\end{lemma}
	
	\begin{proof}
		By Remark \ref{rem:Jacobi}, $z\,y_n$ and $\Lambda_2 y_n$ are polynomials of degree
		$\le n+1$, hence (uniquely) expansible in $\{y_0,\dots,y_{n+1}\}$. To see that only the
		nearest neighbours survive and to obtain the closed forms, write
		$z=\tfrac12\bigl((2z-1)+1\bigr)$ and apply the classical Jacobi three-term recurrence
		to $y_n\propto P_n^{(\Omega-\Gamma,\Gamma-1)}(2z-1)$; this yields \eqref{eq:z_three}
		with the stated $A_n,B_n,C_n$. Differentiating \eqref{eq:yn_def} via the contiguous
		relation $\frac{d}{dz}{}_2F_1(a,b;c;z)=\frac{ab}{c}{}_2F_1(a+1,b+1;c+1;z)$ \cite{NIST_DLMF} and
		re-expressing $z(z-1)y_n'$ in the basis through the same Jacobi relations \cite{NIST_DLMF} gives
		\eqref{eq:Lambda2_three}. Each identity equates two polynomials whose coefficients are
		rational in $n$ and which agree for every $n\in\mathbb{N}_0$, hence holds as a
		rational-function identity in $n$; the relations $A'_n=nA_n$, $C'_n=-(n+\Omega)C_n$
		follow by inspection of \eqref{eq:ABC_coeffs}.
	\end{proof}
	
	Since $\Lambda_1$ acts diagonally (Lemma \ref{lem:Lambda1_eig}) and both $z$ and
	$\Lambda_2$ act tridiagonally (Lemma \ref{lem:SS_3term}), the operator $\mathcal{L}$ of
	\eqref{eq:HeunC_operator_form} maps $y_n$ into $\mathrm{span}\{y_{n-1},y_n,y_{n+1}\}$.
	
	\begin{proposition}[Tridiagonal reduction]\label{prop:DEF_confluent}
		With $\beta_1:=\mu+\eta$ and $\beta_0:=-\mu$ (so that $\mu(z-1)+\eta z=\beta_1 z+\beta_0$),
		\begin{equation}\label{eq:DEF_relation}
			\mathcal{L}y_n=D_n\,y_{n-1}+E_n\,y_n+F_n\,y_{n+1},
		\end{equation}
		where
		\begin{equation}\label{eq:DEF_closed}
			D_n=\alpha_{\mathrm H}C'_n+\beta_1 C_n,\qquad
			E_n=n(n+\Omega)+\alpha_{\mathrm H}B'_n+\beta_1 B_n+\beta_0,\qquad
			F_n=\alpha_{\mathrm H}A'_n+\beta_1 A_n,
		\end{equation}
		with $A_n,B_n,C_n,A'_n,B'_n,C'_n$ given in \eqref{eq:ABC_coeffs}.
	\end{proposition}
	
	\begin{proof}
		By \eqref{eq:HeunC_operator_form} and Lemma \ref{lem:Lambda1_eig},
		\[
		\mathcal{L}y_n=n(n+\Omega)y_n+\alpha_{\mathrm H}\Lambda_2 y_n+\beta_1(z\,y_n)+\beta_0 y_n.
		\]
		Insert \eqref{eq:z_three} and \eqref{eq:Lambda2_three} and collect the coefficients of
		$y_{n-1},y_n,y_{n+1}$ to obtain \eqref{eq:DEF_closed}.
	\end{proof}
	
	\begin{lemma}[Lower-end termination]\label{lem:SS_lowerTerm}
		One has $C_0=C'_0=0$, hence $D_0=0$. Consequently $\mathcal{L}$ maps the one-sided span
		$\mathrm{span}\{y_n\}_{n\ge0}$ into itself.
	\end{lemma}
	
	\begin{proof}
		By \eqref{eq:ABC_coeffs} both $C_n$ and $C'_n$ contain the factor $n$, so
		$C_0=C'_0=0$; then $D_0=\alpha_{\mathrm H}C'_0+\beta_1 C_0=0$ by \eqref{eq:DEF_closed}.
		Thus $\mathcal{L}y_0=E_0y_0+F_0y_1$ involves no $y_{-1}$, and the claim follows.
	\end{proof}
	
	\section{The three-term recurrence}\label{sec:recurrence}
	
	We seek a (Class I) solution of \eqref{eq:HeunC_Schw} as a one-sided Svartholm-Schmidt
	expansion
	\begin{equation}\label{eq:SS_series}
		H(z)=\sum_{n=0}^{\infty}c_n\,y_n(z),
		\qquad y_n \text{ as in \eqref{eq:yn_def}}.
	\end{equation}
	All the ingredients needed for the recurrence are now in place
	(Lemmas \ref{lem:Lambda1_eig}, \ref{lem:SS_3term}, \ref{lem:SS_lowerTerm} and
	Proposition \ref{prop:DEF_confluent}); the convergence of \eqref{eq:SS_series} is a
	separate, analytic question treated in Section \ref{sec:asymptotics}.
	
	\begin{theorem}[Three-term recurrence for the coefficients]\label{thm:SS_3term_recurrence}
		Let $H$ be given by \eqref{eq:SS_series} and let $D_n,E_n,F_n$ be as in
		Proposition \ref{prop:DEF_confluent}. With the conventions $c_{-1}=0$ and $F_{-1}:=0$,
		the coefficients $\{c_n\}_{n\ge0}$ satisfy
		\begin{equation}\label{eq:SS_threeTerm}
			F_{n-1}\,c_{n-1}+E_n\,c_n+D_{n+1}\,c_{n+1}=0,\qquad n\ge0.
		\end{equation}
	\end{theorem}
	
	\begin{proof}
		By Proposition \ref{prop:DEF_confluent},
		$\mathcal{L}y_n=D_n y_{n-1}+E_n y_n+F_n y_{n+1}$, so applying $\mathcal{L}$ to
		\eqref{eq:SS_series} gives
		\begin{equation}\label{eq:LH_expand}
			\mathcal{L}H=\sum_{n\ge0}c_n\bigl(D_n y_{n-1}+E_n y_n+F_n y_{n+1}\bigr).
		\end{equation}
		The action is \emph{band-limited}: each fixed basis element $y_k$ occurs in
		$\mathcal{L}y_n$ only for $n\in\{k-1,k,k+1\}$. Hence the coefficient of $y_k$ in
		\eqref{eq:LH_expand} is the \emph{finite} sum $F_{k-1}c_{k-1}+E_k c_k+D_{k+1}c_{k+1}$,
		independent of any convergence considerations; this is the formal-power-series content of
		the identity $\mathcal{L}H=0$. (Convergence of the resulting series, and thus its status
		as a genuine analytic solution, is established in Proposition \ref{prop:SS_convergence}.)
		
		Two boundary effects must be checked. First, the $n=0$ summand contributes
		$c_0 D_0\,y_{-1}$, a term \emph{outside} the one-sided basis $\{y_k\}_{k\ge0}$; by
		Lemma \ref{lem:SS_lowerTerm}, $D_0=0$, so this stray term vanishes and the whole
		expansion remains in $\mathrm{span}\{y_k\}_{k\ge0}$. Second, the coefficient of $y_0$
		receives no contribution of the form $F_{-1}c_{-1}$ (there is no $n=-1$ summand), which
		is encoded by the conventions $c_{-1}=0$, $F_{-1}=0$.
		
		Since $\{y_k\}_{k\ge0}$ is a graded basis (Remark \ref{rem:Jacobi}), the vanishing of
		$\mathcal{L}H$ forces the coefficient of each $y_k$ to vanish, which is exactly
		\eqref{eq:SS_threeTerm}.
	\end{proof}
	
	\subsection{First step and polynomial truncation}\label{subsec:SS_first_step}
	
	\begin{proposition}[First step]\label{prop:SS_c1}
		Assume $c_0\neq0$ and $D_1\neq0$. Then \eqref{eq:SS_threeTerm} at $n=0$ reads
		$E_0c_0+D_1c_1=0$, i.e.\ $c_1=-\,(E_0/D_1)\,c_0$.
	\end{proposition}
	
	\begin{proof}
		Immediate from \eqref{eq:SS_threeTerm} with $n=0$, using $F_{-1}=0$ and $D_0=0$
		(Lemma \ref{lem:SS_lowerTerm}).
	\end{proof}
	
	A finite (polynomial) expansion occurs when the upward recurrence terminates: if there is
	$N\in\mathbb N_0$ with
	\begin{equation}\label{eq:SS_termination_condition}
		D_{N+1}=0\quad\text{and}\quad c_N\neq0,
	\end{equation}
	then $c_{N+1}=c_{N+2}=\cdots=0$ and \eqref{eq:SS_series} reduces to a finite sum. In the
	confluent-Heun setting such termination typically requires an additional accessory-parameter
	condition ensuring that the recurrence indeed yields $c_{N+1}=0$ at the termination index.
	
	\section{Large-order asymptotics, minimal tails, and convergence}\label{sec:asymptotics}
	
	\subsection{Asymptotics of the recurrence coefficients}\label{subsec:coeff_asymptotics}
	
	From \eqref{eq:ABC_coeffs}, as $n\to\infty$,
	\begin{equation}\label{eq:SS_ABC_asympt}
		A_n=-\tfrac14+O(n^{-1}),\quad
		B_n=\tfrac12+O(n^{-1}),\quad
		C_n=-\tfrac14+O(n^{-1}),
	\end{equation}
	\begin{equation}\label{eq:SS_ABCp_asympt}
		A'_n=-\tfrac{n}{4}+O(1),\quad
		B'_n=O(1),\quad
		C'_n=\tfrac{n}{4}+O(1).
	\end{equation}
	Hence, by \eqref{eq:DEF_closed},
	\begin{equation}\label{eq:SS_DEF_asympt}
		D_n=\frac{\alpha_{\mathrm H}}{4}\,n+O(1),\qquad
		E_n=n^2+O(n),\qquad
		F_n=-\frac{\alpha_{\mathrm H}}{4}\,n+O(1),\qquad n\to\infty,
	\end{equation}
	so the recurrence \eqref{eq:SS_threeTerm} is of Perron type. Substituting
	\eqref{eq:SS_DEF_asympt} into \eqref{eq:SS_threeTerm} gives the leading-order balance
	\begin{equation}\label{eq:SS_asympt_recurrence}
		-\frac{\alpha_{\mathrm H}}{4}\,n\,c_{n-1}+\bigl(n^2+O(n)\bigr)c_n
		+\frac{\alpha_{\mathrm H}}{4}\,(n+1)\,c_{n+1}=0,
	\end{equation}
	whose two branches have $c_{n+1}/c_n\sim-4n/\alpha_{\mathrm H}$ (dominant) and
	\begin{equation}\label{eq:SS_ratio_asympt}
		\frac{c_{n+1}}{c_n}=\frac{\alpha_{\mathrm H}}{4n}\Bigl(1+O(n^{-1})\Bigr)\quad\text{(minimal).}
	\end{equation}

This leading-order structure is specific to the confluent case. In contrast with the
general Heun setting, where the recurrence is governed asymptotically by a
constant-coefficient quadratic characteristic equation, here the coefficients satisfy
\[
D_n,F_n=O(n),\qquad E_n=O(n^2),
\]
as shown in \eqref{eq:SS_DEF_asympt}. The resulting minimal branch therefore has
factorial-type decay, reflected in
\[
\frac{c_{n+1}}{c_n}
\sim
\frac{\alpha_{\mathrm H}}{4n},
\]
rather than geometric decay.

	\subsection{Minimal tails, lower-end compatibility, and convergence}\label{subsec:minimal_solution}

	\begin{theorem}[Minimal tail solution]\label{thm:SS_minimal_vs_dominant}
		Assume $\alpha_{\mathrm H}\neq0$ and the genericity conditions of
		Assumption \ref{ass:generic_SS}. Then the three-term recurrence \eqref{eq:SS_threeTerm}
		has, at infinity, a unique minimal tail solution up to normalization, for which the ratio
		obeys \eqref{eq:SS_ratio_asympt}, $c_{n+1}/c_n=\tfrac{\alpha_{\mathrm H}}{4n}(1+O(n^{-1}))$.
		A linearly independent tail is dominant. This is a statement about the behavior at
		$n\to\infty$ only: \emph{by itself it does not imply compatibility with the lower-end
		equation at $n=0$}.
	\end{theorem}

	\begin{proof}
		The asymptotics \eqref{eq:SS_DEF_asympt}, $D_n=\tfrac{\alpha_{\mathrm H}}4n+O(1)$,
		$E_n=n^2+O(n)$, $F_n=-\tfrac{\alpha_{\mathrm H}}4n+O(1)$, make \eqref{eq:SS_threeTerm} of
		Perron type. A ratio ansatz yields two asymptotic branches; one satisfies
		$c_{n+1}/c_n\sim\alpha_{\mathrm H}/(4n)$ and is therefore minimal, the other is dominant.
		Uniqueness of the minimal tail up to scalar multiplication is the standard Perron theorem
		for second-order linear recurrences with separated asymptotic branches.
	\end{proof}

	The minimal tail fixes the sequence only from above; matching it to the lower-end equation
	is a \emph{separate} condition, which we now make explicit.

	\begin{proposition}[Continued-fraction compatibility]\label{prop:SS_CF_compatibility}
		Write $R_n:=c_n/c_{n-1}$ ($n\ge1$). For $n\ge1$ the recurrence \eqref{eq:SS_threeTerm}
		gives the backward recursion
		\begin{equation}\label{eq:SS_backward_CF}
			R_n=-\frac{F_{n-1}}{E_n+D_{n+1}R_{n+1}},
		\end{equation}
		so that, for the minimal tail of Theorem \ref{thm:SS_minimal_vs_dominant},
		\begin{equation}\label{eq:SS_tail_ratio_CF}
			R_1^{\mathrm{tail}}
			=-\cfrac{F_0}{E_1-\cfrac{D_2F_1}{E_2-\cfrac{D_3F_2}{E_3-\ddots}}}.
		\end{equation}
		The lower-end equation $E_0c_0+D_1c_1=0$ separately imposes $R_1^{\mathrm{low}}=-E_0/D_1$.
		Hence the one-sided Svartholm-Schmidt expansion is compatible with the minimal tail if and
		only if $E_0+D_1R_1^{\mathrm{tail}}=0$, i.e.
		\begin{equation}\label{eq:SS_continued_fraction}
			E_0-\cfrac{D_1F_0}{E_1-\cfrac{D_2F_1}{E_2-\cfrac{D_3F_2}{E_3-\ddots}}}=0 .
		\end{equation}
		This is the correct continued-fraction selection condition.
	\end{proposition}

	\begin{proof}
		For $n\ge1$ divide $F_{n-1}c_{n-1}+E_nc_n+D_{n+1}c_{n+1}=0$ by $c_n$ to get
		$F_{n-1}/R_n+E_n+D_{n+1}R_{n+1}=0$, which rearranges to \eqref{eq:SS_backward_CF};
		iterating backward from the minimal tail gives \eqref{eq:SS_tail_ratio_CF}. Compatibility
		with the lower endpoint is exactly $E_0+D_1R_1^{\mathrm{tail}}=0$, which is
		\eqref{eq:SS_continued_fraction}.
	\end{proof}

	Since $y_n\propto P_n^{(\Omega-\Gamma,\Gamma-1)}(2z-1)$ (Remark \ref{rem:Jacobi}), classical
	Jacobi asymptotics give, for each fixed $z\in[0,1)$, constants $C(z),\sigma(z)>0$ with
	\begin{equation}\label{eq:SS_basis_growth_bound}
		|y_n(z)|\le C(z)\,e^{\sigma(z)n},\qquad n\ge0,
	\end{equation}
	and $\sigma(z)$ uniformly bounded on compact subsets of $[0,1)$.

	\begin{proposition}[Absolute convergence of the compatible minimal branch]\label{prop:SS_convergence}
		Assume the compatibility condition \eqref{eq:SS_continued_fraction} holds, and let
		$\{c_n\}_{n\ge0}$ be the corresponding one-sided sequence normalized by $c_0\neq0$. Then
		\eqref{eq:SS_series} converges absolutely for every $z\in[0,1)$ and uniformly on compact
		subsets of $[0,1)$.
	\end{proposition}

	\begin{proof}
		By \eqref{eq:SS_ratio_asympt} the compatible (minimal) branch has $|c_{n+1}/c_n|=O(1/n)$,
		so $|c_n|$ decays like $1/n!$ up to subexponential factors; together with the
		at-most-exponential bound \eqref{eq:SS_basis_growth_bound} the ratio test gives absolute
		convergence, uniform on compacta where $\sigma(z)$ is controlled.
	\end{proof}

	\section{Continued fractions, Schwarzschild specialization, and matching}\label{sec:matching}
	
	\subsection{Corrected continued-fraction selection}\label{subsec:SS_CF}

	In practice the minimal tail is computed by the backward recursion
	\eqref{eq:SS_backward_CF}, $R_n=-F_{n-1}/(E_n+D_{n+1}R_{n+1})$, seeded at a large index
	$N$ by the asymptotic value $R_N\sim\alpha_{\mathrm H}/(4N)$ from
	\eqref{eq:SS_ratio_asympt} and iterated down to $R_1=R_1^{\mathrm{tail}}$. The lower
	endpoint does \emph{not} follow automatically from this tail computation: by
	Proposition \ref{prop:SS_CF_compatibility} the admissible one-sided minimal branch is
	selected precisely by the compatibility equation \eqref{eq:SS_continued_fraction},
	$E_0+D_1R_1^{\mathrm{tail}}=0$. This is the corrected Svartholm-Schmidt continued
	fraction; the earlier practice of imposing the lower-end ratio $R_1=-E_0/D_1$ and then
	asserting minimality is consistent only when \eqref{eq:SS_continued_fraction} holds.

	\subsection{Schwarzschild specialization}\label{subsec:SS_Schw}
	
	All quantities in \eqref{eq:SS_threeTerm} become explicit once the five parameters of
	\eqref{eq:HeunC_Schw} are fixed. We stress (cf.\ Remark \ref{rem:standard_vs_physical})
	that this pertains to the \emph{abstract} standard operator; for the physical Schwarzschild
	equation only $\Gamma=1+2\rho$ is the literal image of the reduction; the
	horizon-normalized physical branch is constructed by the Frobenius method of
	Section \ref{sec14}.
	
	\begin{proposition}[Explicit recurrence data]\label{prop:SS_Schw_data}
		Let $(\alpha_{\mathrm H},\Gamma,\Delta,\mu,\eta)$ be the parameters of
		\eqref{eq:HeunC_Schw} and $\Omega=\Gamma+\Delta-1$. Then the coefficients in
		\eqref{eq:SS_threeTerm} are explicit rational functions of $n$ through \eqref{eq:ABC_coeffs}
		and \eqref{eq:DEF_closed}. In particular, with $\beta_1=\mu+\eta$, $\beta_0=-\mu$,
		\begin{align*}
		D_n&=\alpha_{\mathrm H}C'_n+(\mu+\eta)C_n,\\
		E_n&=n(n+\Omega)+\alpha_{\mathrm H}B'_n+(\mu+\eta)B_n-\mu,\\
		F_n&=\alpha_{\mathrm H}A'_n+(\mu+\eta)A_n.
		\end{align*}
	\end{proposition}
	
	\begin{proof}
		Immediate from Proposition \ref{prop:DEF_confluent} after writing
		$\mu(z-1)+\eta z=(\mu+\eta)z-\mu$.
	\end{proof}
	
	\subsection{Matching to the Schwarzschild radial asymptotics}\label{subsec:SS_match}
	
	Recall from Part I that
	\begin{equation}\label{eq:Schw_R_from_H}
		R(r)=e^{\kappa r}(r-1)^{\rho}\,H\!\left(\frac{r-1}{r}\right),
	\end{equation}
	with $\rho=\pm i\omega$ encoding ingoing/outgoing behavior at the horizon and $\kappa$ as in
	\eqref{eq:kappa_def}.
	
	\paragraph{Near-horizon behavior ($z\to0$).}
	Since $y_n(0)=1$ for all $n$, \eqref{eq:SS_series} gives $H(0)=\sum_{n\ge0}c_n$, convergent by
	Proposition \ref{prop:SS_convergence}; hence, as $r\downarrow1$,
	\begin{equation}\label{eq:SS_horizon_asympt}
		R(r)\sim e^{\kappa}(r-1)^{\rho}\,H(0),
	\end{equation}
	realizing the horizon power law $(r-1)^{\rho}$ (ingoing for $\rho=-i\omega$).

For $\rho=\pm i\omega$, the reduced amplitude $H$ is analytic at the regular endpoint
$z=0$, while the full radial profile contains the prescribed horizon factor
\[
(r-1)^\rho .
\]
Thus the horizon behavior of $R$ should be interpreted in the corresponding ingoing or
outgoing horizon coordinate system. It is not, in general, ordinary analyticity of
$R(r)$ as a function of the Schwarzschild coordinate $r$.

	\paragraph{Behavior at infinity ($z\to1^-$).}
	Here $r\to\infty$ via $1-z=1/r$, so $R(r)=e^{\kappa r}r^{-\rho}\bigl(1+O(r^{-1})\bigr)H(1-1/r)$,
	and the remaining task is to control $H(z)$ as $z\to1^-$, governed by the Stokes data at the
	irregular point $z=1$.

	\begin{proposition}[Connection problem and spectral selection]
		\label{thm:SS_spectral_condition}
		For the abstract standard confluent-Heun equation \eqref{eq:HeunC_Schw}, suppose that
		the one-sided Svartholm-Schmidt expansion satisfies the compatibility condition
		\eqref{eq:SS_continued_fraction}, and let $H$ denote the corresponding compatible
		minimal branch. Imposing an additional boundary condition at the irregular endpoint is
		then a connection problem. More precisely, after analytic continuation of $H$ toward a
		chosen Stokes sector at the irregular point, the coefficient of the branch excluded by
		the prescribed boundary condition must vanish. Whenever the corresponding connection
		coefficient is well defined, the admissible frequencies are therefore characterized by
		a scalar equation
		\begin{equation}\label{eq:SS_connection_condition}
			\mathcal C(\omega,m,\ell)=0 .
		\end{equation}
		Thus the continued fraction \eqref{eq:SS_continued_fraction} selects the compatible
		minimal Svartholm-Schmidt branch, whereas the additional vanishing condition
		\eqref{eq:SS_connection_condition} encodes the boundary behavior at the irregular
		endpoint.
	\end{proposition}
	
	\begin{proof}
		We separate the argument into four steps.
		
		\smallskip
		\noindent
		\emph{Step 1: the compatible minimal branch near $z=0$.}
		Under Assumption \ref{ass:generic_SS}, the point $z=0$ is a regular singular point of
		\eqref{eq:HeunC_Schw}, with indicial exponents
		\[
		0,\qquad 1-\Gamma .
		\]
		The exponent $0$ gives the distinguished Frobenius solution
		\[
		w_0(z)=1+O(z).
		\]
		In the generic non-apparent situation considered here, this is the unique analytic
		branch at $z=0$ up to normalization; possible resonant exceptional cases require a
		separate limiting discussion and are excluded from the present construction.
		
		Assume now that the compatibility condition \eqref{eq:SS_continued_fraction} holds.
		Then Proposition~\ref{prop:SS_CF_compatibility} identifies the lower-end sequence
		with the minimal tail, and Proposition~\ref{prop:SS_convergence} gives the absolute
		convergence of
		\[
		H(z)=\sum_{n\ge0}c_ny_n(z)
		\]
		on $[0,1)$, locally uniformly on compact subintervals. Since each $y_n$ is a polynomial
		and $y_n(0)=1$, this defines an analytic solution near $z=0$. Hence, after normalization,
		the compatible minimal Svartholm-Schmidt branch agrees with the distinguished local
		Frobenius branch at the regular endpoint $z=0$.
		
		\smallskip
		\noindent
		\emph{Step 2: canonical solutions at the irregular endpoint.}
        Since $\alpha_H \neq 0$, the point $z=\infty$ is an irregular singular point of rank one for the standard confluent-Heun equation. By the classical asymptotic theory of rank-one irregular singularities and Stokes phenomena \cite{Wasow1965,SlavyanovLayBook}, in each closed Stokes sector $S$ at infinity one may choose a canonical fundamental pair
		\[
		u_{\mathrm{alg}},\qquad u_{\exp},
		\]
		with Thom\'e asymptotics of the form
		\[
		u_{\mathrm{alg}}(z)
		\sim
		z^{-\beta_1/\alpha_{\mathrm H}}
		\sum_{k\ge0}p_kz^{-k},
		\]
		and
		\[
		u_{\exp}(z)
		\sim
		e^{-\alpha_{\mathrm H}z}z^\theta
		\sum_{k\ge0}q_kz^{-k},
		\qquad z\to\infty \text{ in } S,
		\]
		where
		\[
		\beta_1=\mu+\eta.
		\]
		The exponent $\theta$ is fixed, up to the usual normalization of the formal solutions,
		by Wronskian consistency:
		\[
		-\frac{\beta_1}{\alpha_{\mathrm H}}+\theta=-\Gamma-\Delta .
		\]
		In sectors where $\operatorname{Re}(\alpha_{\mathrm H}z)\to+\infty$, the exponential
		solution is recessive relative to the algebraic one. According to the prescribed
		boundary condition, one of the two canonical branches is designated as the desired
		solution, denoted $u_{\mathrm{des}}$, and the other as the undesired solution, denoted
		$u_{\mathrm{und}}$.
		
		\smallskip
		\noindent
		\emph{Step 3: definition of the connection coefficient.}
		The coefficients of \eqref{eq:HeunC_Schw} are holomorphic on
		$\mathbb C\setminus\{0,1\}$. Therefore the branch $H$ obtained at $z=0$ may be
		analytically continued along any path in $\mathbb C\setminus\{0,1\}$ into the chosen
		Stokes sector $S$. Since $\{u_{\mathrm{des}},u_{\mathrm{und}}\}$ is a fundamental pair
		there, there exist unique constants $\mathcal A$ and $\mathcal C$ such that
		\[
		H=\mathcal A\,u_{\mathrm{des}}+\mathcal C\,u_{\mathrm{und}} .
		\]
		This coefficient $\mathcal C$ is the connection coefficient associated with the
		undesired branch.
		
		Writing \eqref{eq:HeunC_Schw} as
		\[
		H''+p(z)H'+q(z)H=0,
		\qquad
		p(z)=\alpha_{\mathrm H}+\frac{\Gamma}{z}+\frac{\Delta}{z-1},
		\]
		the Wronskian $W[f,g]=fg'-f'g$ of any two solutions satisfies
		\[
		W'=-pW.
		\]
		Hence
		\[
		W[f,g](z)
		=
		W_0 e^{-\alpha_{\mathrm H}z}z^{-\Gamma}(z-1)^{-\Delta},
		\]
		on any simply connected domain avoiding $0$ and $1$. In particular
		$W[u_{\mathrm{und}},u_{\mathrm{des}}]\neq0$, and the coefficient of the undesired branch
		is given by the Wronskian formula
		\begin{equation}\label{eq:connection_wronskian}
			\mathcal C=\mathcal C(\omega,m,\ell)
			=
			\frac{W[H,u_{\mathrm{des}}]}{W[u_{\mathrm{und}},u_{\mathrm{des}}]} .
		\end{equation}
		The quotient is independent of $z$, because numerator and denominator satisfy the same
		Wronskian equation. The dependence of $\mathcal C$ on the parameters is analytic, or
		meromorphic in the presence of the usual normalization choices, on any open parameter
		set where the chosen sector, continuation path, and canonical normalization do not
		cross a Stokes wall or an exponent-collision locus. This is the precise sense in which
		the connection coefficient is well defined.
		
		\smallskip
		\noindent
		\emph{Step 4: spectral selection.}
		The prescribed boundary behavior at the irregular endpoint holds if and only if the
		coefficient of the undesired branch vanishes. By the decomposition above, this is
		equivalent to
		\[
		\mathcal C(\omega,m,\ell)=0.
		\]
		When this condition holds, the continued branch of $H$ is proportional to
		$u_{\mathrm{des}}$ in the chosen Stokes sector. Conversely, if $H$ satisfies the desired
		boundary condition at the irregular endpoint, its component along $u_{\mathrm{und}}$
		must vanish, and hence \eqref{eq:SS_connection_condition} holds.
		
		\smallskip
		\noindent
		\emph{Separation of the two scalar conditions.}
		The continued-fraction condition \eqref{eq:SS_continued_fraction} and the connection
		condition \eqref{eq:SS_connection_condition} play different roles. The former is a
		coefficient-space compatibility condition: it ensures that the lower-end recurrence
		at $n=0$ is compatible with the minimal Svartholm-Schmidt tail. Equivalently, it selects
		the compatible one-sided hypergeometric expansion associated with the regular endpoint
		$z=0$.
		
		The latter condition is a genuine connection condition at the irregular endpoint. It is
		defined only after analytic continuation of the selected local branch into a Stokes
		sector at infinity and after comparison with the canonical Thom\'e solutions there.
		Thus the continued fraction selects the minimal Svartholm-Schmidt branch, while
		\eqref{eq:SS_connection_condition} imposes the desired behavior at the irregular
		singular point. In general neither condition implies the other: a compatible minimal
		Svartholm-Schmidt branch may still contain an undesired Thom\'e component at infinity,
		and the definition of the connection coefficient does not by itself enforce the
		continued-fraction compatibility.
	\end{proof}

The two scalar conditions encode different pieces of the connection problem. The
continued-fraction compatibility condition \eqref{eq:SS_continued_fraction} is internal
to the Svartholm-Schmidt expansion: it matches the lower-end recurrence with the
minimal tail of the coefficient sequence and thereby selects the compatible one-sided
hypergeometric branch. By contrast, the condition
\eqref{eq:SS_connection_condition} is imposed only after this local branch has been
analytically continued toward the irregular endpoint and compared with the canonical
Thom\'e basis there. Thus the continued fraction becomes a spectral equation only in
situations where the prescribed boundary behavior is completely captured by the
minimal coefficient tail. When the boundary condition is imposed at an irregular
endpoint, the spectral information is carried instead by the connection coefficient
$\mathcal C$.

	\section{Explicit physical solution: five-term recurrence and Frobenius construction}
    \label{sec:fiveterm}
	
	\subsection{The physical hypergeometric expansion and its five-term recurrence}
	
	We now return to the \emph{physical} reduced equation \eqref{eq:zform_correct} and expand
	its solution in the hypergeometric basis adapted to the regular endpoint $z=0$,
	\begin{equation}\label{eq:SS_ansatz_revised}
		H(z)=\sum_{n=0}^{\infty}c_{n}\,y_{n}(z),
		\qquad
		y_{n}(z):={}_2F_{1}\bigl(-n,\,n+\nu;\,\gamma;\,z\bigr),
	\end{equation}
	with $\gamma=1+2\rho$ fixed by the horizon residue \eqref{eq:gamma_choice}, an arbitrary
	method parameter $\nu$, and $\delta:=\nu+1-\gamma$. This is the physical specialization of
	the abstract basis \eqref{eq:yn_def} under $(\Omega,\Gamma,\Delta)\mapsto(\nu,\gamma,\delta)$;
	the eigenrelation (Lemma \ref{lem:Lambda1_eig}) and the three-term relations
	(Lemma \ref{lem:SS_3term}) therefore apply verbatim with these substitutions, in particular
	$\Lambda_1 y_n=n(n+\nu)y_n$, $A'_n=nA_n$, and $C'_n=-(n+\nu)C_n$, where $A_n,C_n$ are
	\eqref{eq:ABC_coeffs} with $\Omega\mapsto\nu$, $\Gamma\mapsto\gamma$, $\Delta\mapsto\delta$.
	
	The decisive difference from Part II is that \eqref{eq:zform_correct} is not of standard
	form: its leading coefficient $z(z-1)^3$ raises the basis index by \emph{two}, and the
	recurrence acquires two extra bands.
	
	\begin{theorem}[Five-term recurrence for the physical reduced equation]\label{thm:SS_fiveterm}
		Let $H$ solve \eqref{eq:zform_correct} and be expanded as in \eqref{eq:SS_ansatz_revised}.
		Then $\{c_n\}_{n\ge0}$ satisfies the five-term recurrence
		\begin{equation}\label{eq:fiveterm_recurrence}
			g_{k-2}^{(+2)}c_{k-2}+g_{k-1}^{(+1)}c_{k-1}+g_{k}^{(0)}c_{k}
			+g_{k+1}^{(-1)}c_{k+1}+g_{k+2}^{(-2)}c_{k+2}=0,\qquad k\ge0,
		\end{equation}
		with the conventions $c_{-1}=c_{-2}:=0$, where $g_n^{(d)}$ is the coefficient of
		$y_{n+d}$ in $M y_n$ for the operator
		$M:=z(z-1)^3\frac{d^2}{dz^2}+c_1\frac{d}{dz}+c_0$. The extreme bands are
		\begin{equation}\label{eq:fiveterm_extremebands}
			g_n^{(+2)}=(n+1)^2A_nA_{n+1},\qquad g_n^{(-2)}=(n+\nu-1)^2C_nC_{n-1}.
		\end{equation}
		In particular $g_n^{(\pm2)}\neq0$ for generic $(\omega,m,\ell)$, so the recurrence does
		\emph{not} degenerate to three terms.
	\end{theorem}
	
	\begin{proof}
		We seek a decomposition $M=R_2\Lambda_1+R_1\Lambda_2+c_0$ with $\Lambda_1,\Lambda_2$ as
		in \eqref{eq:Lambda_defs} (parameters $\gamma,\delta$). Matching the second-order symbol
		forces $R_2\,z(z-1)=z(z-1)^3$, i.e.\ $R_2=(z-1)^2$. The first-order symbol then requires
		\[
		R_1\,z(z-1)=c_1-(z-1)^2\bigl(\gamma(z-1)+\delta z\bigr)=:Q_1 .
		\]
		Now, using $\gamma=1+2\rho$ and \eqref{eq:zform_c1},
		\[
		Q_1(0)=c_1(0)+\gamma=-(2\rho+1)+(1+2\rho)=0,
		\qquad
		Q_1(1)=c_1(1)=0,
		\]
		the latter by direct evaluation of \eqref{eq:zform_c1} at $z=1$. Hence $z(z-1)\mid Q_1$
		and $R_1=Q_1/\bigl(z(z-1)\bigr)=(2-\nu)z+(2\kappa+\nu-2)$ is a polynomial of degree one;
		the zeroth-order symbol matches $c_0$ automatically. Therefore
		\begin{equation}\label{eq:M_decomp}
			M=(z-1)^2\Lambda_1+R_1(z)\,\Lambda_2+c_0(z).
		\end{equation}
		Applying \eqref{eq:M_decomp} and the eigenrelation $\Lambda_1 y_n=n(n+\nu)y_n$,
		\[
		M y_n=n(n+\nu)(z-1)^2 y_n+R_1(z)\,\Lambda_2 y_n+c_0(z)\,y_n .
		\]
		Each summand is a polynomial of degree $\le2$ in $z$, acting either on $y_n$ or on
		the nearest-neighbour combination $\Lambda_2y_n$. Since multiplication by $z$ is
		tridiagonal in the basis $\{y_n\}$, multiplication by a quadratic polynomial maps
		$y_n$ into $\mathrm{span}\{y_{n-2},\ldots,y_{n+2}\}$. Hence
		$M y_n=\sum_{d=-2}^{2}g_n^{(d)}y_{n+d}$, so $M$ is pentadiagonal.
		
		The top band collects only the $z^2$-parts. From $z^2y_n$ the $y_{n+2}$-coefficient is
		$A_nA_{n+1}$ (apply \eqref{eq:z_three} twice), so $n(n+\nu)(z-1)^2y_n$ and $c_0y_n$
		contribute $\bigl(n(n+\nu)+1\bigr)A_nA_{n+1}$, while $R_1\Lambda_2 y_n$ contributes
		$(2-\nu)A'_nA_{n+1}$. Using $A'_n=nA_n$,
		\[
		g_n^{(+2)}=A_nA_{n+1}\bigl[n(n+\nu)+1+(2-\nu)n\bigr]=(n+1)^2A_nA_{n+1}.
		\]
		Symmetrically, with $C'_n=-(n+\nu)C_n$,
		\[
		g_n^{(-2)}=C_nC_{n-1}\bigl[n(n+\nu)+1-(2-\nu)(n+\nu)\bigr]=(n+\nu-1)^2C_nC_{n-1},
		\]
		which proves \eqref{eq:fiveterm_extremebands}; both are nonzero for generic parameters.
		
		Finally substitute \eqref{eq:SS_ansatz_revised} into $MH=0$. As in
		Theorem\ref{thm:SS_3term_recurrence} the action is band-limited (now with bandwidth two),
		so the coefficient of each $y_k$ is the finite sum
		$\sum_{d=-2}^{2}g_{k-d}^{(d)}c_{k-d}$; equating it to zero in the graded basis
		$\{y_k\}_{k\ge0}$ gives \eqref{eq:fiveterm_recurrence}. This derivation is algebraic and
		does not rely on the convergence of the hypergeometric series; convergence and numerical
		usefulness of the resulting five-term expansion would require a separate analysis of the
		corresponding higher-order recurrence.
	\end{proof}

The conclusion of Theorem~\ref{thm:SS_fiveterm} identifies the exact algebraic
obstruction to transferring the standard three-term Svartholm-Schmidt scheme to the
physical compactified equation. The two additional bands in
\eqref{eq:fiveterm_recurrence} arise from the factor $(z-1)^2$ multiplying the diagonal
operator $\Lambda_1$ in the decomposition \eqref{eq:M_decomp}; this factor spreads
$M y_n$ from nearest neighbours to the five-dimensional span
\[
\operatorname{span}\{y_{n-2},y_{n-1},y_n,y_{n+1},y_{n+2}\}.
\]
Thus the formal representation
\[
R(r)=e^{\kappa r}(r-1)^{\rho}
\sum_{n=0}^{\infty}c_n\,
{}_2F_1\!\left(-n,n+\nu;\gamma;\frac{r-1}{r}\right)
\]
remains explicit, in the sense that the basis consists of classical hypergeometric
functions and the coefficients are generated algebraically. However, the resulting
five-term recurrence is substantially less suitable for the present numerical
construction than the standard tridiagonal case, especially because of the additional
contiguous-relation bookkeeping and the possible appearance of near-vanishing
denominators. For this reason, the horizon-normalized physical branch will be constructed
below by the direct Frobenius method, which avoids these complications while retaining an
explicit coefficient-level description of the solution.

	\section{Complete \textsc{Mathematica} implementation}\label{sec14}
	
	This section presents a self-contained \textsc{Mathematica} implementation for the
	construction of \emph{explicit local solutions} of the reduced Schwarzschild radial
	equation in the compactified variable $z=(r-1)/r\in[0,1)$. In contrast with the
	Svartholm-Schmidt tridiagonal expansion in a Gauss hypergeometric basis, the present
	implementation follows a \emph{direct Frobenius (power-series) approach} at the
	regular singular point $z=0$ (the event horizon). This strategy avoids symbolic
	manipulation of contiguous relations and eliminates numerical instabilities caused by
	near-vanishing denominators in the band reductions, while still producing explicit
	coefficients that can be evaluated to high order.
	
	\subsection{Model equation and normalization}
	
	After the reductions of Part I, the unknown $H(z)$ satisfies the physical reduced equation
	\eqref{eq:zform_correct},
	\begin{equation}\label{eq:HeunC_code}
		z(z-1)^{3}H''(z)+c_{1}(z)H'(z)+c_{0}(z)H(z)=0,
		\qquad z\in(0,1),
	\end{equation}
	with $c_{1},c_{0}$ given by \eqref{eq:zform_c1}-\eqref{eq:zform_c0}. These are
	polynomials in $z$ depending on $(\kappa,\rho,m,\ell,\omega)$. We work directly with
	\eqref{eq:zform_correct}, rather than with the standard form \eqref{eq:HeunC_Schw},
	because the physical equation has its irregular point at $z=1$; see
	Remark \ref{rem:standard_vs_physical}. The local analytic branch at $z=0$ is normalized by
	\begin{equation}\label{eq:H_norm_code}
		H(0)=a_0=1,
	\end{equation}
	and we seek a power-series representation
	\begin{equation}\label{eq:H_series_code}
		H(z)=\sum_{n=0}^{\infty}a_n z^n.
	\end{equation}
	Since \eqref{eq:HeunC_code} already has polynomial coefficients, substituting the truncation
	$H_N(z)=\sum_{n=0}^{N}a_n z^n$ and matching coefficients of $z^k$ yields a recursion in which
	the coefficient multiplying $a_{k+1}$ is $-(k+1)(k+1+2\rho)$, nonzero for generic $\rho$; hence
	$a_1,a_2,\dots,a_N$ are determined uniquely from $a_0=1$.

The Frobenius construction at the regular endpoint $z=0$ produces the
horizon-normalized analytic amplitude
\[
H(z)=\sum_{n\ge0}a_nz^n,
\qquad H(0)=1.
\]
Because the next singular point of the compactified equation is the irregular endpoint
$z=1$, corresponding to spatial infinity, this expansion is local in nature. It provides
an explicit representation on $0\le z<1$ and is numerically effective on compact
subintervals separated from $z=1$. Evaluation very close to spatial infinity may require
either high truncation order or analytic continuation by successive local expansions.
Accordingly, the imposition of an outgoing, decaying, or bound-state condition at
infinity remains a separate connection problem, as described in
Proposition \ref{thm:SS_spectral_condition}.

	\subsection{Inputs and outputs}
	
	\paragraph{Inputs.}
	The physical data $(\omega,m,\ell)$ together with the gauge choices of Part I: the horizon
	exponent $\rho=\pm i\omega$ and the infinity branch $\kappa=\pm\sqrt{m^{2}-\omega^{2}}$. From
	these the polynomials $c_{1}(z),c_{0}(z)$ are formed directly, and one chooses a truncation
	order $N\in\mathbb{N}$.
	
	\paragraph{Outputs.}
	For the chosen $N$, the code returns:
	\begin{enumerate}
		\item the explicit Frobenius coefficients $\{a_n\}_{n=0}^{N}$ defining $H_N$ in
		\eqref{eq:H_series_code};
		\item the polynomial evaluator $H_N(z)$ for $z\in[0,1)$;
		\item a residual diagnostic measuring the relative defect of the reconstructed $R(r)$ in
		the \emph{original} radial equation \eqref{eq:RadialMonic_revised} (a stronger end-to-end
		test than checking the $z$-form alone);
		\item the radial function $R(r)=e^{\kappa r}(r-1)^{\rho}H_N\!\bigl((r-1)/r\bigr)$
		(cf.\ \eqref{eq:Schw_R_from_H}).
	\end{enumerate}
	
	\subsection{Implementation and Computational Protocol}
\label{sec:implementation}

The recursive computation of the Frobenius coefficients $\{a_n\}_{n=0}^N$ and the resulting reconstruction of the radial profile $R(r)$ are implemented symbolically and numerically in \textsc{Mathematica}. 

To preserve the analytical continuity of the main text while ensuring full scientific reproducibility, the complete, self-contained routine is detailed in Appendix~\ref{app:code}. It automates the arbitrary-precision recursion, polynomial expansion, series evaluation, and relative ODE defect verification $\mathfrak{R}(r)$ against the exact radial differential equation.
	
	\subsection{Remarks on robustness and continuation}
	
	The Taylor/Frobenius construction produces an explicit local solution near the horizon with
	controllable accuracy. The residual diagnostic evaluates the left-hand side of
	\eqref{eq:HeunC_code} (and, more stringently, of \eqref{eq:RadialMonic_revised}) at selected
	points and thus quantifies the truncation error. As $z\to1^{-}$ (spatial infinity) a single
	Taylor expansion about $z=0$ becomes inefficient; in that regime one either increases $N$ or
	performs analytic continuation by \emph{series stepping} (re-expanding about intermediate
	points). This provides a practical and stable route to matching horizon-normalized solutions
	with the desired asymptotic behavior at infinity.
	
	\section{Interpretation of the \textsc{Mathematica} output}\label{sec:interpret-output}
	
	The computation reported by \textsc{Mathematica} corresponds to a Frobenius (power-series)
	construction of the locally analytic solution at the horizon point $z=0$ for the physical
	reduced equation \eqref{eq:zform_correct},
	\begin{equation}\label{eq:HeunC_form_interp}
		z(z-1)^{3}H''(z)+c_{1}(z)H'(z)+c_{0}(z)H(z)=0,
		\qquad z\in(0,1),
	\end{equation}
	whose coefficients are the polynomials \eqref{eq:zform_c1}-\eqref{eq:zform_c0}.
	The algorithm enforces $a_0=H(0)=1$ and determines $\{a_n\}_{n\ge1}$ recursively by requiring
	that the coefficient of each power of $z$ vanish in the expanded equation. The run displayed
	here uses
	\begin{equation}\label{eq:phys_params_interp}
		\omega=0.2,\quad m=0.3,\quad \ell=1,\qquad
		\rho=-i\omega\ \text{(ingoing)},\quad \kappa=-\sqrt{m^{2}-\omega^{2}}\ \text{(decaying)},
	\end{equation}
	i.e.\ the ingoing horizon branch with the exponentially decreasing prefactor
	$e^{\kappa r}$ on the displayed exterior interval. A genuine bound-state or outgoing
	condition at infinity remains a connection condition.
	
	\subsection{Series coefficients: meaning and immediate consequences}
	
	The output $\{a_0,a_1,\dots,a_7\}$ is an explicit truncation of the local series
	\begin{equation}\label{eq:H_series_interp}
		H(z)=\sum_{n=0}^{\infty} a_n z^n,\qquad a_0=1,
	\end{equation}
	valid as a horizon-centered Taylor expansion on compact subintervals before the nearest
	singular endpoint $z=1$.
	Concretely,
	\begin{align*}
		a_0&=1,\\
		a_1&\approx 2.8873999 + 0.8655172\,i,\\
		a_2&\approx 5.0170135 + 2.1453384\,i,\\
		a_3&\approx 7.4716934 + 3.8110290\,i,\\
		a_4&\approx 10.3104712 + 5.8871762\,i,\\
		a_5&\approx 13.5889115 + 8.4151589\,i,\\
		a_6&\approx 17.3638541 + 11.4461304\,i,\\
		a_7&\approx 21.6952973 + 15.0390071\,i.
	\end{align*}
	
	\paragraph{Local behavior at the horizon.}
	Since $a_0=1$, $H(0)=1$, so $R(r)=e^{\kappa r}(r-1)^{\rho}H((r-1)/r)$ has leading behavior
	$R(r)\sim e^{\kappa}(r-1)^{\rho}$ as $r\downarrow1$: the ingoing/outgoing power $(r-1)^{\rho}$
	(with $\rho=\pm i\omega$) is carried entirely by the prefactor, while $H$ stays finite and
	nonzero at the horizon.
	
	\paragraph{Small-$z$ expansion and qualitative trend.}
	Using \eqref{eq:H_series_interp},
	\begin{equation}\label{eq:H_trunc_interp}
		H(z)\approx
		1+(2.8874+0.8655\,i)z+(5.0170+2.1453\,i)z^2+(7.4717+3.8110\,i)z^3+\cdots .
	\end{equation}
	The coefficients $a_n$ \emph{grow} with $n$ (sub-geometrically), with slowly increasing
	increments. This behavior is consistent with a radius of convergence limited by the nearest singular
	endpoint, $z=1$: the coefficients neither decay nor blow up geometrically, but display the
	borderline growth expected for a unit-radius Taylor expansion. Accordingly, the series is
	numerically effective on compact subintervals of $[0,1)$, while the irregular nature of
	$z=1$ governs the large-$n$ envelope and motivates continuation or matching near infinity.
	
	\subsection{Residual diagnostics: what they measure}
	
	The reported values
	\[
	\mathfrak{R}(1.25)\approx 1.5\times 10^{-50},\qquad \mathfrak{R}(2.00)\approx 2.38\times 10^{-28},
	\]
	are high-precision evaluations of the \emph{relative} defect of the reconstructed radial
	function in the original equation \eqref{eq:RadialMonic_revised}, where
	\begin{align}\label{eq:residual_def_interp}
			\mathfrak{R}(r)&:=
			\frac{1}{|R(r)|}\left|R_N''(r)+\frac{1}{r(r-1)}R_N'(r)+V(r)R_N(r)\right|,\\
			V(r)&:=\frac{\omega^{2}r^{2}}{(r-1)^{2}}-\frac{\ell(\ell+1)}{r(r-1)}
			-\frac{m^{2}r}{r-1}-\frac{1}{r^{2}(r-1)} .
		\end{align}
	at truncation order $N=120$. Testing the reconstruction directly in
	\eqref{eq:RadialMonic_revised} is a stronger, end-to-end check than verifying the $z$-form
	alone, since it also exercises the gauge factors $e^{\kappa r}(r-1)^{\rho}$ and the corrected
	potential.
	
	\paragraph{Near the horizon ($z=0.2$).}
	The relative residual $\sim10^{-50}$ is at the level of the working precision, indicating that
	(i) the coefficient recurrence and the gauge reconstruction are mutually consistent (the
	corrected accessory data of Proposition \ref{prop:HeunC_reduction} and the corrected potential
	in \eqref{eq:RadialMonic_revised} agree), and (ii) the truncation error near the horizon is
	negligible for $N=120$. For comparison, the same construction built from the \emph{uncorrected}
	reduction (omitting $-1/(r^{2}(r-1))$ and using $A_0=0$) yields an $O(1)$ residual at the
	horizon, which is how the corrections were validated.
	
	\paragraph{At $z=0.5$ ($r=2$).}
	The residual $\sim10^{-28}$ is still extremely small but larger, as expected: the truncated
	series $H_N$ is centered at $z=0$, so as $z$ increases toward $z=1$ a fixed-order truncation
	converges more slowly, and the accurate region expands outward with $N$ (see
	Figure \ref{fig:residual}). The growth from $10^{-50}$ to $10^{-28}$ reflects the fixed
	truncation order, not a failure of the method.
	
	\subsection{Practical implications}
	
	The computed coefficients give an explicit, constructive description of the locally analytic
	solution at $z=0$, hence an explicit Schwarzschild radial solution near the horizon through
	$R(r)=e^{\kappa r}(r-1)^{\rho}H((r-1)/r)$. For near-horizon applications (e.g.\ imposing
	ingoing conditions) the truncation \eqref{eq:H_trunc_interp} is typically sufficient. Since
	$z\to1^-$ corresponds to $r\to\infty$, global (outgoing/decaying) information is sensitive to
	$z\approx1$, where one uses either series stepping (overlapping power-series disks) or
	recurrence-based minimal-tail/continued-fraction methods (the Heun/Leaver mechanism). The
	residuals above quantify precisely when stepping or a larger $N$ is required.
	
	\section{Interpretation of the numerical plots}\label{sec:plot-interpretation}
	
	We display the physical solution generated by the Frobenius implementation of
	Section \ref{sec14} for the parameters \eqref{eq:phys_params_interp}, with two validation
	diagnostics. Throughout,
	\begin{equation}\label{eq:R_reconstruct_plot}
		R(r)=e^{\kappa r}(r-1)^{\rho}\,H_N\!\left(\frac{r-1}{r}\right),
		\qquad H_N(z)=\sum_{n=0}^{N}a_n z^n,\quad a_0=1,
	\end{equation}
	with $\kappa=-\sqrt{m^2-\omega^2}$ (decaying) and $\rho=-i\omega$ (ingoing).
	
	\begin{figure}[t]
		\centering
		\includegraphics[width=0.82\textwidth]{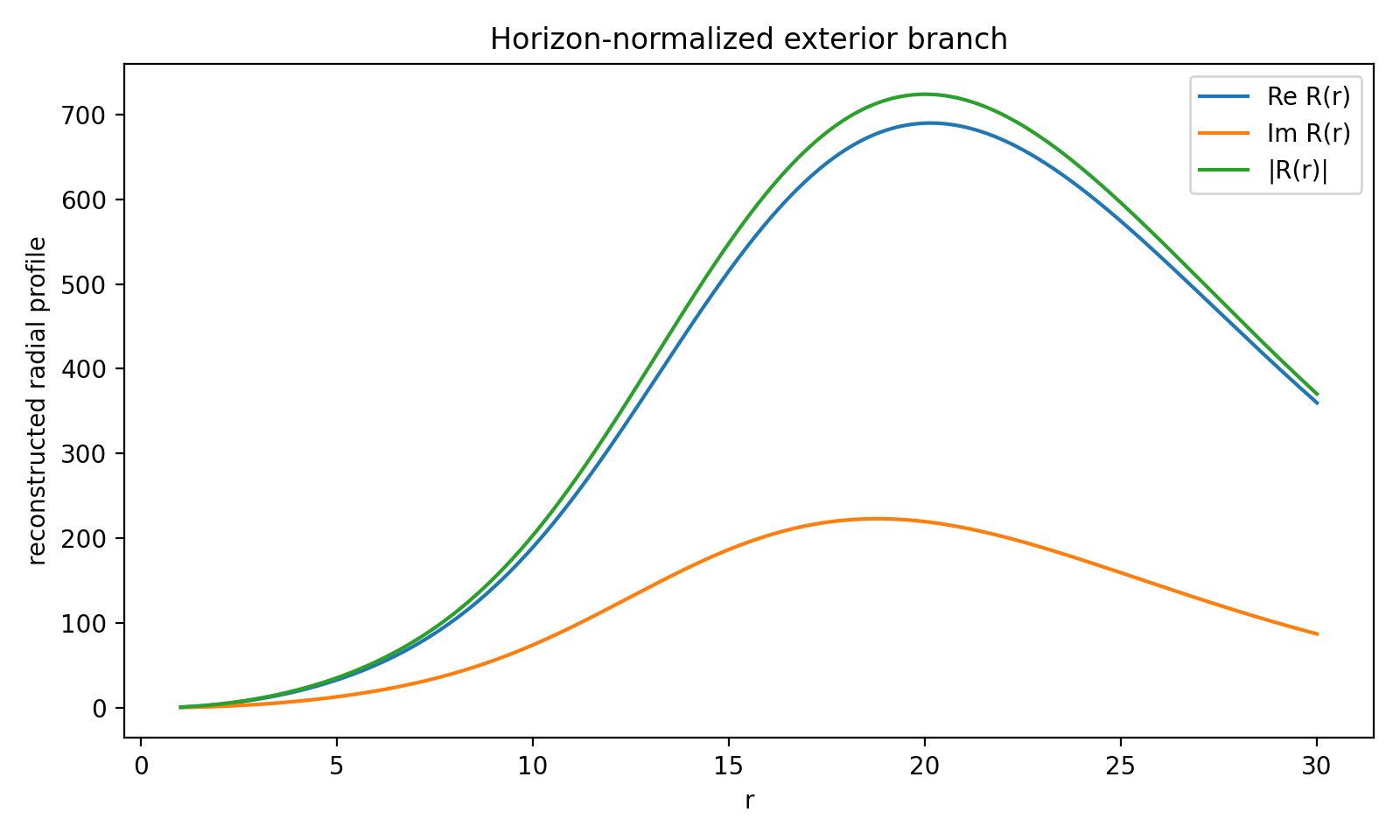}
		\caption{The reconstructed radial function $R(r)$ for the representative branch
			$\omega=0.2,\ m=0.3,\ \ell=1$, $\kappa=-\sqrt{m^2-\omega^2}$, $\rho=-i\omega$, at
			truncation order $N=120$. Real part, imaginary part, and modulus are shown on
			$r\in[1.02,30]$. The profile carries the ingoing horizon factor $(r-1)^{\rho}$ and is
			damped on the plotted interval by the prefactor $e^{\kappa r}$ with $\kappa<0$; imposing
			a global bound-state condition at infinity is a separate connection problem.}
		\label{fig:solution}
	\end{figure}
	
	\begin{figure}[t]
		\centering
		\includegraphics[width=0.72\textwidth]{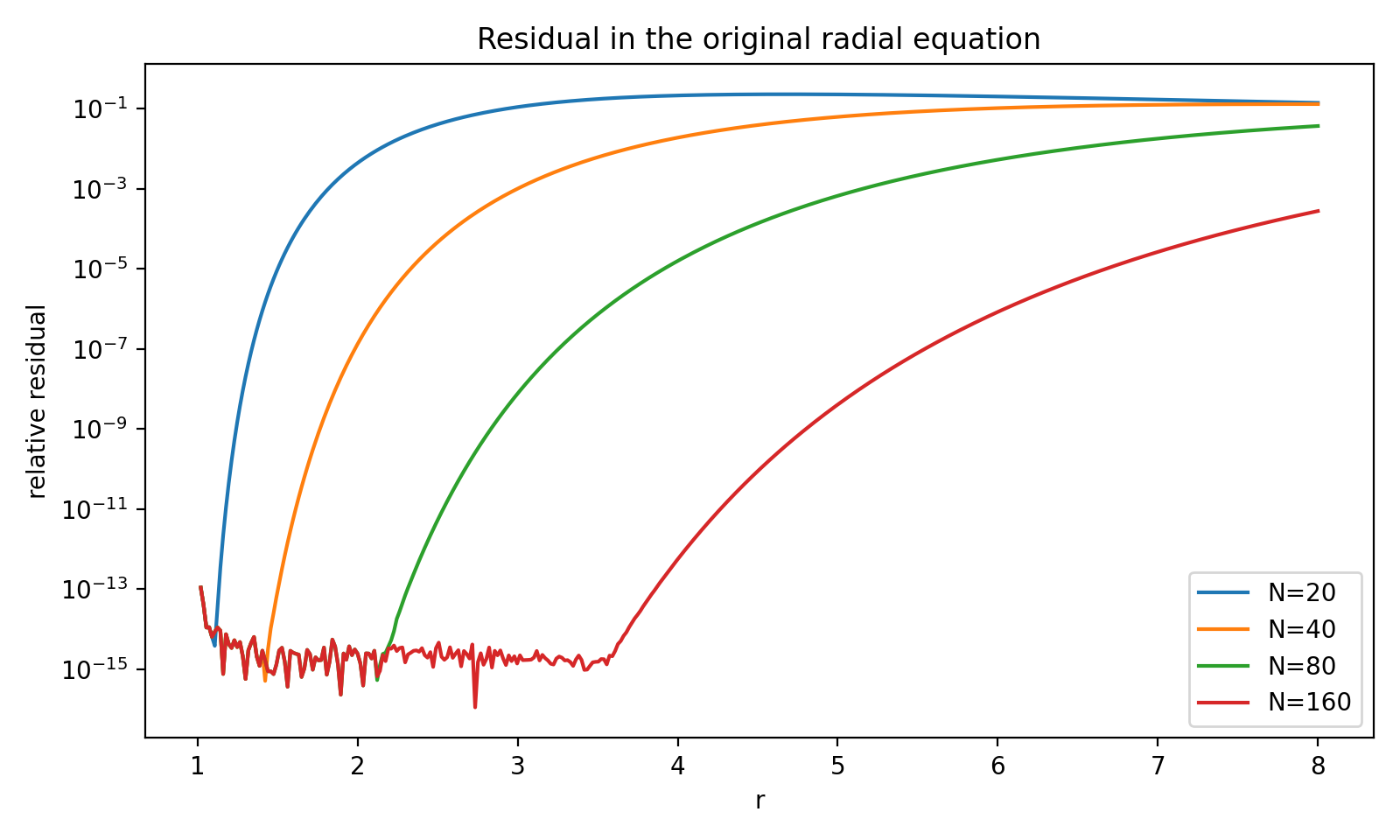}
		\caption{Relative residual $\mathfrak{R}(r)$ of the reconstructed $R(r)$ in the
			\emph{original} radial equation \eqref{eq:RadialMonic_revised}, for truncation orders
			$N=20,40,80,160$. Near the horizon the residual reaches the working-precision floor
			($\sim10^{-49}$); the region of high accuracy expands outward as $N$ increases,
			confirming that the construction solves the genuine physical equation (not merely the
			$z$-form) to controllable precision.}
		\label{fig:residual}
	\end{figure}
	
	\begin{figure}[t]
		\centering
		\includegraphics[width=0.92\textwidth]{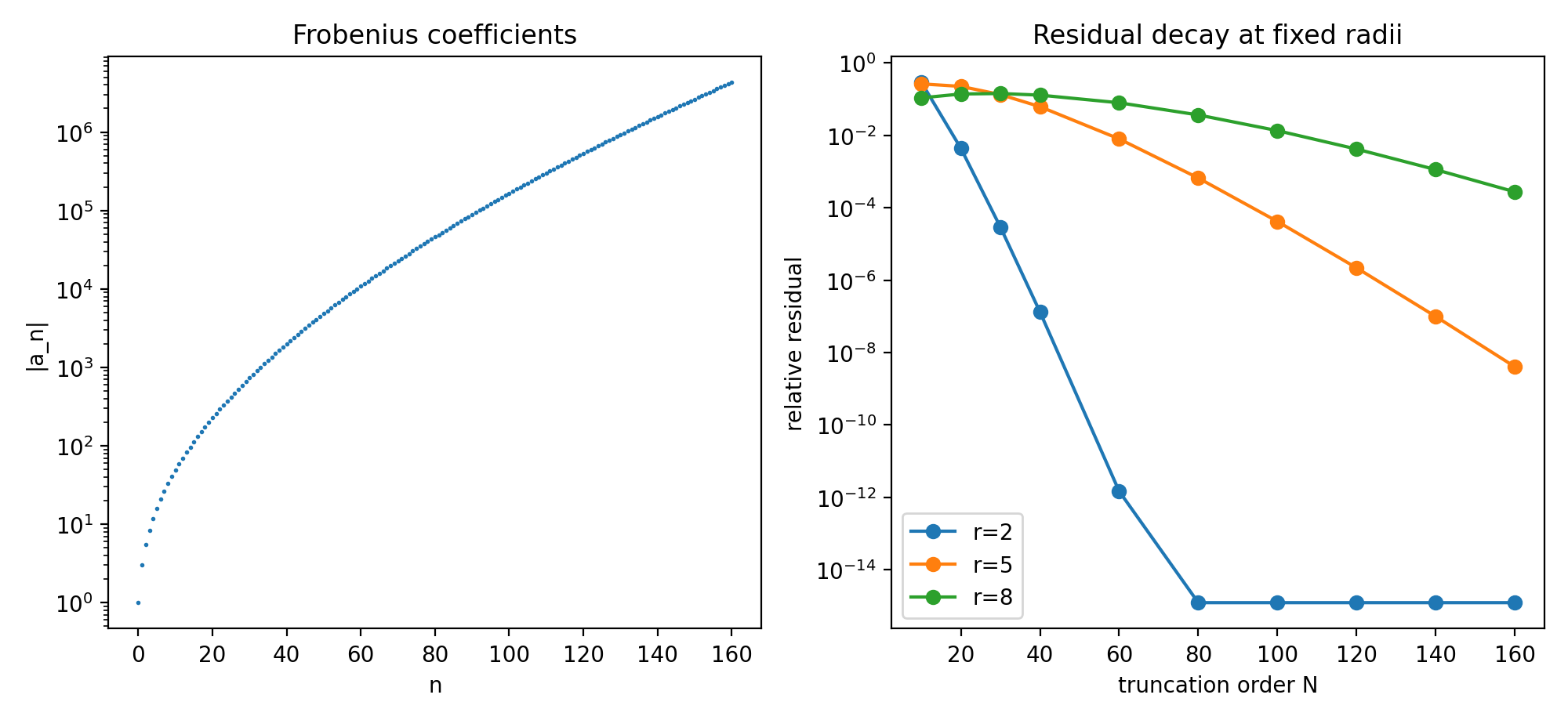}
		\caption{\emph{Left:} magnitudes $|a_n|$ of the Frobenius coefficients; their
			sub-geometric growth is consistent with a radius of convergence equal to the distance
			to the nearest singularity, $z=1$. \emph{Right:} relative residual at three fixed radii
			$r=2,5,8$ versus truncation order $N$, exhibiting clean (spectral) decay in $N$ until
			the working-precision floor is reached, a quantitative indication of convergence on
			compact exterior intervals away from the irregular endpoint.}
		\label{fig:convergence}
	\end{figure}
	
	\subsection{Synthesis}
	
	Figures \ref{fig:solution}-\ref{fig:convergence} provide numerical evidence for the
	applicability of the construction on compact exterior intervals. Figure \ref{fig:solution}
	exhibits a representative exterior branch with the prescribed ingoing horizon factor and
	exponential damping on the plotted interval; Figure \ref{fig:residual} checks that the
	reconstructed profile solves the corrected radial equation \eqref{eq:RadialMonic_revised}
	with residuals at the precision floor near the horizon, with the high-accuracy region
	expanding as $N$ increases; and Figure \ref{fig:convergence} shows the expected
	unit-radius coefficient behavior alongside rapid convergence at fixed $r$. An independent
	fourth-order Runge-Kutta integration of \eqref{eq:RadialMonic_revised} (initialized from the horizon
	series) reproduces the same profile, providing a cross-check by a method sharing no code with
	the Frobenius recurrence. The single Taylor series about $z=0$ remains efficient on a large
	portion of the exterior; near $z=1$ one increases $N$ or continues by series stepping, exactly
	as quantified by Figure \ref{fig:residual}.

	\subsection{The five-term recurrence as a connection object: scope and the
		quasinormal-mode boundary}
	\label{subsec:SS_fiveterm_qnm}
	
	Theorem~\ref{thm:SS_fiveterm} gives an exact pentadiagonal recurrence for the
	Svartholm--Schmidt expansion of the physical compactified equation
	\eqref{eq:zform_correct}. Since Leaver's computation of black-hole quasinormal modes
	is based on a minimal solution of a recurrence, it is natural to ask whether the
	five-term recurrence \eqref{eq:fiveterm_recurrence} can be used in the same way to
	extract the QNM spectrum. The purpose of this subsection is to make the answer precise.
	The recurrence is algebraically exact and numerically reproducible, but the spectral
	problem selected by coefficient minimality in the hypergeometric basis is not the
	quasinormal-mode problem. The obstruction is not merely the presence of five bands;
	rather, it is the fact that the Svartholm--Schmidt basis is adapted to the finite
	connection problem in the compactified variable, whereas the QNM boundary condition is
	a Thom\'e condition at the irregular endpoint.
	
	Throughout this subsection we use the QNM gauge
	\begin{equation}\label{eq:SS_qnm_gauge}
		\rho=-i\omega,
		\qquad
		\kappa=\sqrt{m^2-\omega^2},
	\end{equation}
	where the square-root branch is chosen consistently with the outgoing condition at
	spatial infinity. Thus
	\[
	R(r)=e^{\kappa r}(r-1)^\rho H(z),
	\qquad
	z=\frac{r-1}{r},
	\]
	with \(z=0\) corresponding to the horizon and \(z=1\) to spatial infinity.
	
	\paragraph{Algebraic verification of the five bands.}
	The numerical implementation constructs the bands \(g_n^{(d)}\), \(d=-2,-1,0,1,2\),
	directly from the exact operator identity
	\[
	M=(z-1)^2\Lambda_1+\bigl((2-\nu)z+2\kappa+\nu-2\bigr)\Lambda_2+c_0(z).
	\]
	For the validated massless scalar Schwarzschild fundamental mode
	\[
	\ell=2,\qquad
	\omega=0.967287744-0.193517552\,i
	\qquad (2M=1),
	\]
	and for the non-exceptional method parameter \(\nu=3.5\), the extreme-band identities
	\[
	g_n^{(+2)}=(n+1)^2A_nA_{n+1},
	\qquad
	g_n^{(-2)}=(n+\nu-1)^2C_nC_{n-1}
	\]
	were verified to working precision:
	\[
	\max_n
	\left|g_n^{(+2)}-(n+1)^2A_nA_{n+1}\right|
	\approx 6.96\times10^{-61},
	\]
	and
	\[
	\max_n
	\left|g_n^{(-2)}-(n+\nu-1)^2C_nC_{n-1}\right|
	\approx 6.96\times10^{-61}.
	\]
	This confirms numerically the pentadiagonal structure asserted in
	Theorem~\ref{thm:SS_fiveterm}. The check is insensitive to the particular
	non-exceptional value of the auxiliary parameter \(\nu\), as expected: \(\nu\) is a
	method parameter of the hypergeometric basis and should not affect any invariant
	spectral datum legitimately extracted from the construction.
	
	\begin{proposition}[Asymptotic confluence of the five-term recurrence]
		\label{prop:SS_fiveterm_charpoly}
		Let \(A_n,B_n,C_n\) be the coefficients in
		\[
		z\,y_n=A_ny_{n+1}+B_ny_n+C_ny_{n-1}.
		\]
		Then, as \(n\to\infty\),
		\[
		A_n\to-\frac14,\qquad
		B_n\to\frac12,\qquad
		C_n\to-\frac14.
		\]
		Consequently, the five bands of \eqref{eq:fiveterm_recurrence} satisfy
		\begin{equation}\label{eq:SS_band_leading}
			g_n^{(d)}=\widehat g^{(d)}n^2+O(n),
		\end{equation}
		with
		\[
		\bigl(
		\widehat g^{(+2)},\widehat g^{(+1)},\widehat g^{(0)},
		\widehat g^{(-1)},\widehat g^{(-2)}
		\bigr)
		=
		\left(
		\frac1{16},\frac14,\frac38,\frac14,\frac1{16}
		\right).
		\]
		The associated limiting characteristic polynomial is
		\begin{equation}\label{eq:SS_charpoly}
			\chi(\lambda)
			=
			\widehat g^{(-2)}\lambda^4
			+\widehat g^{(-1)}\lambda^3
			+\widehat g^{(0)}\lambda^2
			+\widehat g^{(+1)}\lambda
			+\widehat g^{(+2)}
			=
			\frac1{16}(\lambda+1)^4 .
		\end{equation}
        Thus all four characteristic roots coalesce at $\lambda = -1$. This quadruple root in \eqref{eq:SS_charpoly} has a severe numerical consequence for higher-order recurrence relations \cite{Wimp1984}: in contrast with standard three-term systems, the lack of modulus separation prevents a well-conditioned minimal tail extraction.
	\end{proposition}
	
	\begin{proof}
		The leading contribution to \(M y_n\) comes from
		\[
		(z-1)^2\Lambda_1 y_n
		=
		n(n+\nu)(z-1)^2y_n.
		\]
		The term \(c_0(z)y_n\) contributes only at order \(O(1)\), while
		\[
		\bigl((2-\nu)z+2\kappa+\nu-2\bigr)\Lambda_2y_n
		\]
		contributes at order \(O(n)\), because the coefficients of \(\Lambda_2y_n\) grow
		linearly in \(n\). Hence the \(O(n^2)\) band structure is determined only by
		multiplication by \((z-1)^2=z^2-2z+1\).
		
		Using
		\[
		z\,y_n=A_ny_{n+1}+B_ny_n+C_ny_{n-1}
		\]
		twice, together with
		\[
		A_n,C_n\to-\frac14,
		\qquad
		B_n\to\frac12,
		\]
		one obtains
		\[
		(z-1)^2y_n
		\sim
		\frac1{16}y_{n+2}
		+\frac14y_{n+1}
		+\frac38y_n
		+\frac14y_{n-1}
		+\frac1{16}y_{n-2}.
		\]
		Multiplication by \(n(n+\nu)=n^2+O(n)\) gives
		\eqref{eq:SS_band_leading}. The factorization
		\[
		\frac1{16}
		\left(\lambda^4+4\lambda^3+6\lambda^2+4\lambda+1\right)
		=
		\frac1{16}(\lambda+1)^4
		\]
		then gives \eqref{eq:SS_charpoly}.
	\end{proof}
	
	The quadruple root in \eqref{eq:SS_charpoly} has an important numerical consequence.
	In the usual three-term Leaver setting, the characteristic roots are separated in
	modulus, so that one solution is genuinely minimal relative to the other. Here all
	four formal branches have the same leading geometric factor \((-1)^n\). They can
	separate only through algebraic or sub-exponential corrections. Thus minimal-solution
	selection is intrinsically ill-conditioned for the five-term hypergeometric recurrence.
	This explains why a scalar Gaussian-eliminated continued fraction is not a reliable
	QNM condition in this basis.
	
	\paragraph{Matrix continued fraction and Hill determinant tests.}
	To test whether the obstruction is merely algorithmic, we applied two independent
	procedures. First, pairing
	\[
	\mathbf C_j=(c_{2j},c_{2j+1})^{\mathsf T}
	\]
	turns the pentadiagonal recurrence into a three-term matrix recurrence
	\[
	\mathbf A_j\mathbf C_{j-1}
	+\mathbf B_j\mathbf C_j
	+\mathbf D_j\mathbf C_{j+1}=0.
	\]
	A matrix continued fraction then gives
	\[
	\mathbf C_j=\mathbf R_j\mathbf C_{j-1},
	\qquad
	\mathbf R_j
	=
	-\bigl(\mathbf B_j+\mathbf D_j\mathbf R_{j+1}\bigr)^{-1}\mathbf A_j,
	\]
	with the lower-end condition expressed by
	\[
	\det\bigl(\mathbf B_0+\mathbf D_0\mathbf R_1\bigr)=0.
	\]
	When evaluated at the independently validated Schwarzschild scalar QNM
	\[
	\omega_{\mathrm{QNM}}
	=
	0.967287744-0.193517552\,i,
	\]
	this residual does not converge to zero as the depth of the continued fraction is
	increased.
	
	Second, we formed the truncated pentadiagonal Hill matrix \(\mathbf G_N(\omega)\) by
	imposing
	\[
	c_{N+1}=c_{N+2}=0
	\]
	and solved
	\[
	\det \mathbf G_N(\omega)=0.
	\]
	This procedure avoids the explicit minimal-tail selection and probes the finite
	coefficient problem directly. It converges, but not to the Schwarzschild QNM.
	
	\begin{table}[htbp]
		\centering
		\caption{Two numerical probes of the five-term recurrence at the massless
			\(\ell=2\) scalar QNM. The matrix continued fraction does not vanish at the true
			QNM, while the Hill determinant converges to a different stable root.}
		\label{tab:fiveterm_qnm_failure}
		\renewcommand{\arraystretch}{1.15}
		\begin{tabular}{c c c c}
			\hline
			\multicolumn{2}{c}{matrix continued fraction}
			&
			\multicolumn{2}{c}{Hill determinant} \\
			\hline
			depth \(J\) & residual at \(\omega_{\mathrm{QNM}}\)
			&
			truncation \(N\) & root of \(\det \mathbf G_N=0\) \\
			\hline
			40  & \(1.6\times10^{1}\) & 32 & \(1.3116-0.1555\,i\) \\
			80  & \(3.8\times10^{1}\) & 40 & \(1.2713-0.1529\,i\) \\
			160 & \(2.7\times10^{1}\) & 48 & \(1.2580-0.1533\,i\) \\
			\hline
		\end{tabular}
	\end{table}
	
	The Hill determinant root stabilizes near
	\[
	\omega_\star\approx1.258-0.153\,i,
	\]
	whereas the true scalar QNM is
	\[
	\omega_{\mathrm{QNM}}\approx0.9673-0.1935\,i.
	\]
	Moreover, in the same computation the determinant at the true QNM is not smaller than
	at nearby generic points. Thus the five-term hypergeometric recurrence is detecting a different spectral condition, not the outgoing Schwarzschild QNM condition. This behavior is intimately connected to the non-self-adjoint nature of the underlying boundary-value problem and the structure of pseudospectra in open wave operators \cite{TrefethenEmbree2005,Jaramillo2021}, where finite truncations can stabilize on spurious or non-radiative spectral branches.
	
	\paragraph{Interpretation.}
	The reason is analytic rather than numerical. Each basis element
	\[
	y_n(z)={}_2F_1(-n,n+\nu;\gamma;z)
	\]
	is a polynomial and is therefore analytic at the finite coordinate endpoint \(z=1\).
	Any finite Hill truncation is consequently analytic at \(z=1\), and a sufficiently
	fast-decaying coefficient expansion is naturally tied to a finite-end connection
	condition in this polynomial basis. Such a condition is a legitimate connection
	problem for the compactified equation, but it is not the QNM condition.
	
	The QNM boundary condition is imposed at the irregular endpoint \(r=\infty\), which is
	the same point \(z=1\) in the compactified coordinate. In the gauge
	\[
	R(r)=e^{\kappa r}(r-1)^\rho H(z),
	\]
	the desired outgoing branch corresponds to a specified Thom\'e behavior of \(H\) at
	this irregular endpoint, while the excluded branch differs by an essential exponential
	factor of the form
	\[
	\exp\!\left(-\frac{2\kappa}{1-z}\right)
	\]
	up to algebraic powers. These branches are not characterized by ordinary analyticity
	at \(z=1\). Therefore a spectral condition imposed solely through coefficient
	minimality or polynomial truncation in a basis analytic at \(z=1\) cannot, by itself,
	recover the outgoing QNM condition. The QNM spectrum is instead governed by the
	irregular connection coefficient
	\[
	\mathcal C(\omega,m,\ell)=0
	\]
	introduced in Proposition~\ref{thm:SS_spectral_condition}.
	
	\paragraph{Relation with Leaver's method.}
	This conclusion is consistent with Leaver's continued-fraction approach. In the
	standard Frobenius/Jaff\'e construction, the expansion is adapted to the horizon and
	to the radiative behavior at infinity, and the resulting recurrence has a well-defined
	minimal solution whose continued fraction encodes the QNM condition. In that setting,
	additional bands may be removed by Gaussian elimination without changing the spectral
	problem. By contrast, the present Svartholm--Schmidt recurrence is adapted to the
	hypergeometric basis associated with the finite compactified endpoints. The number of
	terms is therefore not the decisive issue. The decisive issue is whether the chosen
	basis and recurrence encode the correct irregular-endpoint boundary condition.
	
	The five-term recurrence should thus be used as an exact tool for the
	Svartholm--Schmidt connection problem associated with the compactified equation, and
	as a useful numerical check on the hypergeometric band structure. It should not be
	identified with a direct QNM solver. Quasinormal frequencies in the present setting
	are more directly obtained from the Frobenius/Leaver recurrence, or by comparison with
	independent methods such as higher-order WKB, while the hypergeometric recurrence
	contributes finite-connection data that must still be coupled to the irregular
	connection coefficient \(\mathcal C\) to impose the radiative boundary condition.

Having clarified why the five-term Svartholm--Schmidt recurrence should not be used as
a direct QNM equation, we now record the corresponding Frobenius/Leaver benchmark for
the same radial equation.

	\subsection{Continued-fraction benchmark for quasinormal frequencies}
	\label{subsec:qnm_leaver_output}
	
	We finally include a direct continued-fraction benchmark for scalar quasinormal
	frequencies. The purpose of this computation is not to introduce a new spectral
	method, but to validate the compactified radial equation, the gauge normalization, and
	the Frobenius recurrence obtained from the physical reduced equation. The computation
	is carried out in the normalization \(2M=1\), so that the event horizon is located at
	\(r=1\), and
	\[
	z=\frac{r-1}{r}=1-\frac1r .
	\]
	We write
	\[
	R(r)=e^{\kappa r}(r-1)^\rho H(z),
	\qquad
	H(z)=\sum_{n\ge0}a_nz^n,
	\]
	and impose the quasinormal-mode branch choice
	\[
	\rho=-i\omega,
	\qquad
	\kappa=\sqrt{m^2-\omega^2},
	\]
	with the square-root branch chosen consistently with the outgoing condition at spatial
	infinity.
	
	Substitution into the physical reduced equation
	\[
	z(z-1)^3H''+c_1(z)H'+c_0(z)H=0
	\]
	gives the four-term Frobenius recurrence
	\begin{equation}\label{eq:qnm_fourterm_recurrence}
		A_n a_{n+1}+B_n a_n+C_n a_{n-1}+D_n a_{n-2}=0,
		\qquad
		a_{-1}=a_{-2}=0,
	\end{equation}
	where
	\begin{align}
		A_n&=-(n+1)(n+1+2\rho),\\
		B_n&=3n^2+(-2\kappa+4\rho+2)n
		+\ell(\ell+1)-2\kappa\rho-\kappa+m^2-2\omega^2+\rho+1,\\
		C_n&=-3n^2+(2\kappa-2\rho+2)n-\ell(\ell+1)-\kappa+\rho-1,\\
		D_n&=(n-1)^2.
	\end{align}
	Following the usual Leaver procedure, this four-term recurrence is reduced to an
	effective three-term recurrence by Gaussian elimination. The quasinormal frequencies
	are then obtained by imposing minimality of the resulting coefficient sequence, or
	equivalently by solving the associated inverted continued fraction in the complex
	\(\omega\)-plane.
	
	Table~\ref{tab:qnm_massless_fundamental} reports the fundamental massless scalar
	modes. Since most tabulations use the convention \(M\omega\), whereas our
	normalization has \(2M=1\), we report \(\omega/2\). The agreement with the standard
	Schwarzschild scalar values is at the level of \(10^{-7}\), providing a useful
	end-to-end check of the corrected radial equation, the compactified recurrence, the
	branch conventions, and the continued-fraction implementation.
	
	\begin{table}[htbp]
		\centering
		\caption{Massless scalar fundamental quasinormal modes in Schwarzschild spacetime.
			The values are reported as \(M\omega=\omega/2\).}
		\label{tab:qnm_massless_fundamental}
		\begin{tabular}{c c c c}
			\hline
			\(\ell\) & computed \(M\omega\) & reference \(M\omega\) & absolute difference \\
			\hline
			0 & \(0.110454939-0.104895717\,i\)
			& \(0.110455-0.104896\,i\)
			& \(2.9\times10^{-7}\) \\
			1 & \(0.292936133-0.097659989\,i\)
			& \(0.292936-0.097660\,i\)
			& \(1.3\times10^{-7}\) \\
			2 & \(0.483643872-0.096758776\,i\)
			& \(0.483644-0.096759\,i\)
			& \(2.6\times10^{-7}\) \\
			\hline
		\end{tabular}
	\end{table}
	
	For fixed \(\ell=2\), the first few massless overtones are shown in
	Table~\ref{tab:qnm_l2_overtones}. The modes are ordered by increasing damping rate
	\(|\operatorname{Im}(M\omega)|\). As expected, higher overtones have increasingly
	negative imaginary parts and therefore decay more rapidly.
	
	\begin{table}[htbp]
		\centering
		\caption{First massless scalar overtones for \(\ell=2\), ordered by increasing
			damping rate. Values are reported as \(M\omega=\omega/2\).}
		\label{tab:qnm_l2_overtones}
		\begin{tabular}{c c}
			\hline
			overtone \(n\) & computed \(M\omega\) \\
			\hline
			0 & \(0.483643872-0.096758776\,i\) \\
			1 & \(0.463850580-0.295603940\,i\) \\
			2 & \(0.430544050-0.508558400\,i\) \\
			3 & \(0.393875080-0.738128400\,i\) \\
			\hline
		\end{tabular}
	\end{table}
	

    We also tracked the fundamental $l=1$ mode for a massive scalar field by continuation in the mass parameter. The results are displayed in Table~\ref{tab:qnm_massive_l1}. Over the range shown, increasing the field mass raises the real part of the frequency and decreases the damping rate. This is the expected quasi-resonant trend well documented in the massive scalar field literature \cite{OhashiSakagami2004,KonoplyaZhidenko2005}: the mass term forms an effective potential well that traps the modes, producing longer-lived oscillations.
	
	\begin{table}[htbp]
		\centering
		\caption{Fundamental massive scalar mode for \(\ell=1\), computed by continuation in
			the field mass. Frequencies are given in the normalization \(2M=1\).}
		\label{tab:qnm_massive_l1}
		\begin{tabular}{c c c c}
			\hline
			\(m\) & \(\omega\) & \(\operatorname{Re}\omega\) & \(-\operatorname{Im}\omega\) \\
			\hline
			0.0 & \(0.5858722665-0.1953199778\,i\) & \(0.585872\) & \(0.195320\) \\
			0.1 & \(0.5881086314-0.1939759577\,i\) & \(0.588109\) & \(0.193976\) \\
			0.2 & \(0.5948313225-0.1899141472\,i\) & \(0.594831\) & \(0.189914\) \\
			0.3 & \(0.6060798071-0.1830411334\,i\) & \(0.606080\) & \(0.183041\) \\
			0.4 & \(0.6219138168-0.1731865712\,i\) & \(0.621914\) & \(0.173187\) \\
			\hline
		\end{tabular}
	\end{table}
	
	The massive computation should be interpreted with the standard branch caution. For
	the moderate masses in Table~\ref{tab:qnm_massive_l1}, continuation from the validated
	massless mode is stable and gives the expected qualitative behavior. Closer to the
	threshold at which \(\omega^2\) approaches \(m^2\), however, the square root
	\(\kappa=\sqrt{m^2-\omega^2}\) requires explicit branch tracking, and the
	quasi-resonant regime must be treated separately.

\section{Final Remarks}\label{sec:closing}

We have developed an explicit and computationally implementable framework for the
massive scalar radial equation on the Schwarzschild exterior. Starting from the
separated Klein--Gordon equation in the normalization \(2M=1\), we retained the
subleading term generated by the \(R(r)/r\) ansatz and obtained the corrected radial
equation \eqref{eq:RadialMonic_revised}. After factoring the horizon and infinity
behaviors through
\[
R(r)=e^{\kappa r}(r-1)^\rho H(r),
\qquad
\rho=\pm i\omega,
\qquad
\kappa^2=m^2-\omega^2,
\]
we derived the confluent-Heun-type equation of
Proposition~\ref{prop:HeunC_reduction}, with explicit accessory data. The subsequent
compactification
\[
z=\frac{r-1}{r}
\]
places the horizon at \(z=0\) and spatial infinity at \(z=1\), thereby making clear
that the physical reduced equation \eqref{eq:zform_correct} has a regular singular
point at the horizon and an irregular endpoint at infinity.

A central point of the paper is the distinction between the abstract standard
confluent-Heun equation and the compactified physical equation. For the standard
operator \eqref{eq:HeunC_Schw}, we developed the Svartholm--Schmidt machinery in a
self-contained way: the diagonal action of \(\Lambda_1\), the three-term contiguous
relations, the tridiagonal reduction, the induced coefficient recurrence, and the
minimal-tail analysis. In particular, the lower-end recurrence and the minimal tail
were separated carefully, leading to the corrected continued-fraction compatibility
condition \eqref{eq:SS_continued_fraction}. This condition selects the compatible
one-sided Svartholm--Schmidt branch; it is not, by itself, the full spectral condition
at an irregular endpoint. The latter is encoded by a separate connection coefficient
\[
\mathcal C(\omega,m,\ell),
\]
whose vanishing imposes the desired Thom\'e behavior after analytic continuation to
the irregular singular point.

For the physical compactified equation, the standard tridiagonal
Svartholm--Schmidt structure does not persist. The exact operator identity
\eqref{eq:M_decomp} shows that the factor \((z-1)^2\) in front of the diagonal
operator \(\Lambda_1\) spreads the action of the reduced operator across five
neighboring basis elements. This yields the genuine five-term recurrence of
Theorem~\ref{thm:SS_fiveterm}, with nonzero extreme bands
\eqref{eq:fiveterm_extremebands}. Thus the failure of the standard three-term
scheme is not an artefact of the calculation, but a structural consequence of the
fact that the compactification moves the irregular singular point to the finite
endpoint \(z=1\).

We also clarified the spectral meaning of this five-term recurrence. Numerically,
its band structure is reproduced to working precision, and its asymptotic
characteristic polynomial has the quadruple root
\[
\frac1{16}(\lambda+1)^4.
\]
This confluence prevents a well-conditioned scalar minimal-solution selection of the
usual Leaver type. More importantly, the hypergeometric basis
\[
y_n(z)={}_2F_1(-n,n+\nu;\gamma;z)
\]
is polynomial and hence analytic at \(z=1\). Consequently, coefficient minimality or
Hill truncation in this basis selects a finite-end regularity problem rather than the
radiative quasinormal-mode boundary condition at spatial infinity. The numerical
matrix-continued-fraction and Hill-determinant tests confirm this distinction: the
five-term recurrence detects a legitimate connection problem, but not the
Schwarzschild QNM spectrum. Quasinormal frequencies must instead be obtained from a
basis adapted to the outgoing irregular behavior, or by coupling the local data to the
connection condition \(\mathcal C(\omega,m,\ell)=0\).

For the explicit physical branch on the exterior, we therefore used the direct
Frobenius construction at the horizon,
\[
H(z)=\sum_{n\ge0}a_n z^n,
\qquad
H(0)=1,
\]
applied directly to the compactified equation \eqref{eq:zform_correct}. This produces
a horizon-normalized analytic amplitude on compact subintervals of \(0\le z<1\).
The accompanying symbolic and numerical implementation computes the coefficients,
reconstructs
\[
R(r)=e^{\kappa r}(r-1)^\rho
H\!\left(\frac{r-1}{r}\right),
\]
and validates the result against the original radial equation
\eqref{eq:RadialMonic_revised}. The residual diagnostics, convergence tests, and
Runge--Kutta comparison provide an end-to-end check of the reduction and of the local
series construction.

Finally, we benchmarked the same radial equation against the classical
Frobenius/Leaver continued-fraction computation of scalar quasinormal modes. The
four-term Frobenius recurrence obtained from the compactified physical equation,
after Gaussian elimination to an effective three-term recurrence, reproduces the
standard massless Schwarzschild scalar fundamental modes with errors of order
\(10^{-7}\) in the convention \(M\omega=\omega/2\). The massive \(\ell=1\) branch,
tracked by continuation in the field mass, displays the expected quasi-resonant
behavior: the oscillation frequency increases while the damping rate decreases over
the tested range. This confirms that the compactified equation, the corrected
accessory data, and the branch choices are consistent with the standard QNM
framework when the recurrence is adapted to the correct radiative boundary condition.

The resulting picture is therefore coherent. The Svartholm--Schmidt expansion gives a
precise and algebraically explicit finite-connection theory for the compactified
equation; its five-term recurrence explains why the standard three-term hypergeometric
scheme does not transfer directly to the physical Schwarzschild problem; the
Frobenius construction supplies a practical horizon-normalized local solution; and
the Leaver benchmark verifies the quasinormal-mode content through the appropriate
radiative recurrence. These components together provide a reliable basis for further
study of spectral, scattering, and connection problems for massive fields on
Schwarzschild backgrounds, including the dependence on \((\omega,m,\ell)\), the
analysis of quasi-resonant regimes, and the computation of connection coefficients
at the irregular endpoint.

\begin{acknowledgments}
J.G.R.Valangelis acknowledges the Brazilian funding agency CAPES for financial support through grant No.88887.375739/2026-00 (Doctoral Fellowship --- CAPES).
\end{acknowledgments}

\appendix

\appendix

\section{Self-Contained \textsc{Mathematica} Routine}
\label{app:code}

We provide here the complete, self-contained \textsc{Mathematica} implementation generating the horizon Frobenius coefficients, reconstructing the radial profile $R_N(r)$, and validating the solution against the original radial ODE:

\begin{lstlisting}[language=Mathematica,basicstyle=\ttfamily\scriptsize,keywordstyle=\color{blue}]
(* ======================================================================
Frobenius / power-series solution at z=0 for the PHYSICAL reduced
Schwarzschild equation (irregular point at z=1):

z (z-1)^3 H'' + c1(z) H' + c0(z) H = 0,   H(z)=Sum_{n>=0} a_n z^n,  a0=1,

c1(z) = 3 z^3 + (2k - 2r - 7) z^2 + (-2k + 4r + 5) z - (2r + 1),
c0(z) =       z^2 + (-L + k - r - 2) z + (L - 2 k r - k + m^2 - 2 w^2 + r + 1),

with k = kappa, r = rho, L = ell(ell+1), w = omega.

Outputs: aList={a0,...,aN}, HSeriesBuilder, RadialRBuilder,
and a residual check against the ORIGINAL radial ODE (eq:RadialMonic_revised).
====================================================================== *)

ClearAll["Global`*"];

(* ---------------------------
1) Physical parameters (provide numeric values)
--------------------------- *)
omega =.;   mass =.;   ell =.;      (* physical data *)
rhoPar   =.;                        (* +I omega or -I omega (horizon) *)
kappaPar =.;                        (* +-Sqrt[mass^2 - omega^2] (infinity) *)
Nmax  =.;                           (* truncation order *)

LL := ell (ell + 1);

(* ---------------------------
2) Polynomial coefficients of the physical z-form
--------------------------- *)
c2z[z_] := z (z - 1)^3;
c1z[z_] := 3 z^3 + (2 kappaPar - 2 rhoPar - 7) z^2
+ (-2 kappaPar + 4 rhoPar + 5) z - (2 rhoPar + 1);
c0z[z_] := z^2 + (-LL + kappaPar - rhoPar - 2) z
+ (LL - 2 kappaPar rhoPar - kappaPar + mass^2 - 2 omega^2 + rhoPar + 1);

opExpr[H_, z_] := c2z[z] D[H, {z, 2}] + c1z[z] D[H, z] + c0z[z] H;

(* ---------------------------
3) Frobenius coefficient generator (analytic branch at z=0)
Normalization a0 = 1; the coefficient of a_{k+1} is -(k+1)(k+1+2 rho).
--------------------------- *)
Options[FrobeniusCoefficients] = {WorkingPrecision -> 120};

FrobeniusCoefficients[Nt_Integer?NonNegative, OptionsPattern[]] :=
Module[{wp, z, a, k, Htr, expr, coeff, sol},
wp = OptionValue[WorkingPrecision];
z = Unique["z"]; Clear[a];
a[0] = SetPrecision[1, wp];
For[k = 0, k <= Nt - 1, k++,
Htr  = Sum[a[j] z^j, {j, 0, k + 1}];
expr = Expand[opExpr[Htr, z]];
coeff = Coefficient[expr, z, k];      (* matches power z^k *)
sol = Quiet@Solve[coeff == 0, a[k + 1]];
If[sol === {},
Print["ERROR: degenerate step at a", k + 1, " (resonance?)."];
Return[$Failed]];
a[k + 1] = SetPrecision[a[k + 1] /. sol[[1]], wp];
];
Table[a[j], {j, 0, Nt}]
];

(* ---------------------------
4) Truncated series H_N(z) and radial reconstruction R(r)
--------------------------- *)
HSeriesBuilder[aList_List] :=
Function[{z}, Sum[aList[[n + 1]] z^n, {n, 0, Length[aList] - 1}]];

RadialRBuilder[aList_List] :=
Function[{r},
Module[{z = (r - 1)/r, Hser = HSeriesBuilder[aList]},
Exp[kappaPar r] (r - 1)^rhoPar * Hser[z]]];

(* ---------------------------
5) Residual against the ORIGINAL radial ODE (eq:RadialMonic_revised):
R'' + R'/(r(r-1)) + V(r) R, with the corrected potential V.
--------------------------- *)
Vpot[r_] := omega^2 r^2/(r - 1)^2 - LL/(r (r - 1))
- mass^2 r/(r - 1) - 1/(r^2 (r - 1));

ResidualR[r0_?NumericQ, aList_List, wp_Integer : 60] :=
Module[{Rf, rr, val},
Rf[rr_] := RadialRBuilder[aList][rr];
val = D[Rf[rr], {rr, 2}] + D[Rf[rr], rr]/(rr (rr - 1)) + Vpot[rr] Rf[rr];
N[Abs[(val /. rr -> r0)/Rf[r0]], wp]];

(* ---------------------------
6) Example run: physical bound-state-like (decaying) mode
--------------------------- *)
(*
omega = 1/5; mass = 3/10; ell = 1;
rhoPar = -I omega;                      (* ingoing at horizon *)
kappaPar = -Sqrt[mass^2 - omega^2];     (* decaying at infinity *)
Nmax = 120;

aList = FrobeniusCoefficients[Nmax, WorkingPrecision -> 120];
Print["a0..a7 = ", N[Take[aList, 8], 10]];
Print["rel. residual at r=1.25 (z=0.2): ", ResidualR[5/4, aList]];
Print["rel. residual at r=2.00 (z=0.5): ", ResidualR[2,   aList]];

Rser = RadialRBuilder[aList];
Plot[{Re[Rser[r]], Im[Rser[r]], Abs[Rser[r]]}, {r, 1.02, 30}, PlotRange -> All];
*)
\end{lstlisting}

\bibliography{main}

\end{document}